\documentclass[a4,11pt]{article}
\usepackage[margin=1in]{geometry}
\usepackage{multirow}
\usepackage{amsthm}
\usepackage{thmtools}

\declaretheorem{theorem}
\declaretheoremstyle[%
  spaceabove=-6pt,%
  spacebelow=6pt,%
  headfont=\normalfont\itshape,%
  postheadspace=1em,%
  qed=\qedsymbol%
]{mystyle}
\declaretheorem[name={Proof},style=mystyle,unnumbered,
]{prf}

\newtheorem{lemma}{Lemma}

\newcommand{\head}[1]
 {\markright{\hbox to 0pt{\vtop to 0pt{\hbox{}\vskip 3mm \hrule
 width  \textwidth \vss} \hss}{\sc #1}}}
\usepackage{latexsym,graphicx,amsfonts,amsmath}
\begin{document}

%%\begin{titlepage}
\title{\bf On the optimality of exact and approximation algorithms for
scheduling problems\thanks{Part of the work was done when the first
author was visiting Kiel University. Supported in part
by Chinese Scholarship Council, NSFC (11271325) and the DFG, project, Laufzeitschranken f\"ur Scheduling- und Packungsprobleme
unter Annahme der Exponentialzeithypothese, JA 612 /16-1.}}

\author{Lin Chen$^1$\ \ Klaus Jansen$^2$\ \ Guochuan
Zhang$^1$
\\{\small $^1$College of Computer Science, Zhejiang University, Hangzhou, 310027, China
}
\\{\small chenlin198662@zju.edu.cn, zgc@zju.edu.cn}
\\  {\small $^2$ Department of Computer Science, Kiel University, 24098 Kiel,
  Germany }
\\{\small kj@informatik.uni-kiel.de}
 }
\date{}
\maketitle

\begin{abstract}
We consider the classical scheduling problem on parallel identical
machines to minimize the makespan. There is a long history of
studies on this problem, focusing on exact and approximation
algorithms, and it is thus natural to consider whether these
algorithms are optimal in terms of the running time. Under the
exponential time hypothesis (ETH), we achieve the following results
in this paper:
\begin{itemize}
\item The scheduling problem
on a constant number $m$ of identical machines, denoted by
$Pm||C_{max}$, is known to admit a fully polynomial time
approximation scheme (FPTAS) of running time $O(n) +
(1/\epsilon)^{O(m)}$ (indeed, the algorithm works for an even more
general problem where machines are unrelated). We prove this
algorithm is essentially the best possible in the sense that a
$(1/\epsilon)^{O(m^{1-\delta})}+n^{O(1)}$ time PTAS implies that ETH
fails.
\item The scheduling problem
on an arbitrary number of identical machines, denoted by
$P||C_{max}$, is known to admit a polynomial time approximation
scheme (PTAS) of running time
$2^{O(1/\epsilon^2\log^3(1/\epsilon))}+O(n^{O(1)})$. We prove this
algorithm is nearly optimal in the sense that a
$2^{O((1/\epsilon)^{1-\delta})}+n^{O(1)}$ time PTAS for any
$\delta>0$ implies that ETH fails, leaving a small room for
improvement.
\item The traditional dynamic programming algorithm for $P||C_{max}$
is known to run in $2^{O(n)}$ time. We prove this is essentially the
best possible in the sense that even if we restrict that there are
$n$ jobs and the processing time of each job is bounded by $O(n)$,
an exact algorithm of running time $2^{O(n^{1-\delta})}$ for any
$\delta>0$ implies that ETH fails.
\end{itemize}

To obtain our results we will provide two new reductions from 3SAT,
one for $Pm||C_{max}$ and one for $P||C_{max}$. Indeed, the new
reductions explore the structure of scheduling problems and can also
lead to other interesting results. For example, the recent paper of
Bhaskara et al.~\cite{lowrank} consider the minimum makespan
scheduling problem where the matrix of job processing times
$P=(p_{ij})_{m\times n}$ is of a low rank. They prove that rank 4
scheduling is APX-hard while the rank 2 scheduling is not, leaving
the classification of rank 3 scheduling as an open problem. Using
the framework of our reduction for $P||C_{max}$, rank 3 scheduling
is proved to be APX-hard~\cite{note}.

\bigskip
\smallskip\noindent{\bf Keywords:} {Approximation schemes; Scheduling; Lower bounds; Exponential time hypothesis}
\end{abstract}
\thispagestyle{empty}
\newpage
\setcounter{page}{1}

\vspace{-2mm}
\section{Introduction}
\vspace{-3mm}
The complexity theory allows us to rule out polynomial time
algorithms for many fundamental optimization problems under the
assumption $P \not= NP$. On the other hand, however, this does not
give us (non-polynomial) lower bounds on the running time for such
algorithms. For example, under the assumption $P\not= NP$, there
could still be an algorithm with running time $n^{O(\log n)}$ for
3-SAT or bin packing. A stronger assumption, the Exponential Time
Hypothesis (ETH), was introduced by Impagliazzo, Paturi, and Zane
\cite{IPZ2001}:

\noindent{\bf Exponential Time Hypothesis (ETH)}: There is
a positive real $\delta$ such that 3-SAT with $n$ variables and $m$
clauses cannot be solved in time $2^{\delta n} (n+m)^{O(1)}$.

Using the Sparsification Lemma by Impagliazzo et al. \cite{IPZ2001},
the ETH assumption implies that there is no algorithm for 3-SAT with
$n$ variables and $m$ clauses that runs in time $2^{\delta m}
(n+m)^{O(1)}$ for a real $\delta > 0$ as well. Under the ETH
assumption, lower bounds on the running time for several graph
theoretical problems have been obtained via reductions between
decision problems. For example, there is no $2^{\delta n}$ time
algorithm for 3-Coloring, Independent Set, Vertex Cover, and
Hamiltonian Path unless the ETH assumption fails. An essential
property of the underlying strong reductions to show these lower
bounds is that the main parameter, the number of vertices, is
increased only linearly. These lower bounds together with matching
optimal algorithms of running time $2^{O(n)}$ gives us some evidence
that the ETH is true, i.e. that a subexponential time algorithm for
3-SAT is unlikely to exist. For a nice survey about lower bounds via
the ETH we refer to~\cite{marx2}.
Interestingly, using the ETH assumption one can also prove lower
bounds on the running time of approximation schemes. For example,
Marx \cite{marx1} proved that there is no PTAS of running time
$2^{O((1/\epsilon)^{1-\delta})} n^{O(1)}$ for Maximum Independent
Set on planar graphs.

There are only few lower bounds known for scheduling and packing
problems. Chen et al. \cite{chen2006} showed that precedence
constrained scheduling on $m$ machines cannot be solved in time
$f(m) |I|^{o(m)}$ (where $|I|$ is the length of the instance),
unless the parameterized complexity class $W[1] = FPT$. Kulik and
Shachnai \cite{kulik2010} proved that there is no PTAS for the 2D
knapsack problem with running time $f(\epsilon)
|I|^{o(\sqrt{1/\epsilon})}$, unless all problems in SNP are solvable
in sub-exponential time. Patrascu and Williams \cite{pat2010} proved
using the ETH assumption a lower bound of $n^{o(k)}$ for sized
subset sum with $n$ items and cardinality value $k$. Recently,
Jansen et al. \cite{jansen2013} showed a lower bound of $2^{o(n)}
|I|^{O(1)}$ for the subset sum and partition problem and proved that
there is no PTAS for the multiple knapsack and 2D knapsack problem
with running time $2^{o(1/\epsilon)} |I|^{O(1)}$ and
$n^{o(1/\epsilon)} |I|^{O(1)}$, respectively.

In this paper, we consider the classical scheduling problem of jobs
on identical machines with the objective of minimizing the makespan,
i.e., the largest completion time. Formally, an instance $I$ is
given by a set ${\cal M}$ of $m$ identical machines and a set ${\cal
J}$ of $n$ jobs with processing times $p_j$. The objective is to
compute a non-preemptive schedule or an assignment $a: {\cal J} \to
{\cal M}$ such that each job is executed by exactly one machine and
the maximum load $\max_{i=1,\ldots,m} \sum_{j: a(j) = i} p_j$ among
all machines is minimized. In scheduling theory, this problem is
denoted by $Pm||C_{max}$ if $m$ is a constant or $P||C_{max}$ if $m$
is an arbitrary input.

This problem is NP-hard even if $m=2$, and is strongly NP-hard if
$m$ is an input. On the other hand, for any $\epsilon>0$ there is a
$(1+\epsilon)$-approximation algorithm for $Pm||C_{max}$
\cite{sahni} and $P||C_{max}$
\cite{hochbaum}. Furthermore, there is a long history of
improvements on the running time of such algorithms, the reader may refer to
the full version for the literature, and we mention that currently the best
known FPTAS for $Pm||C_{max}$ has a running time
of $O(n)+(1/\epsilon)^{O(m)}$ \cite{jansen unrelated2} for sufficiently small
$\epsilon$ (e.g., $\epsilon < 1/m$), and the best known PTAS for $P||C_{max}$
has a running time of $2^{O(1/\epsilon
\log^2(1/\epsilon))} + n^{O(1)}$ \cite{jansen}.

Exact algorithms for the scheduling problem are also under extensive
research. Recently Lente et al. \cite{kindt} provided algorithms of
running time $2^{n/2}$ and $3^{n/2}$ for $P2||C_{max}$ and
$P3||C_{max}$, respectively. O'Neil \cite{sub-ex,sub neil} gave a
sub-exponential time algorithm of running time $2^{O(m\sqrt{|I|})}$
for the bin packing problem with $m$ bins where $|I|$ is the length of the
input, and it works also for $Pm||C_{max}$.

The main contribution of this paper is to characterize lower bounds
on the running times of exact and approximation algorithms for the
classical scheduling problem. We prove the following.

\vspace{-2mm}
\begin{theorem}
\label{th:id-sche-linear} For any $\delta > 0$, there is no
$2^{O((1/{\epsilon})^{1-\delta})}+n^{O(1)}$ time PTAS for
$P||C_{max}$, unless ETH fails.
\end{theorem}

\begin{theorem}
\label{th:id-sche-number} For any $\delta>0$, there is no
$2^{O(n^{1-\delta})}$ time exact algorithm for $P||C_{max}$ with $n$
jobs even if we restrict that the processing time of each job is
bounded by $O(n)$, unless ETH fails.
\end{theorem}

\begin{theorem}
\label{th:sche-machine parameter} For any $\delta > 0$, there is no
$(1/\epsilon)^{O(m^{1-\delta})}+n^{O(1)}$ time FPTAS for
$Pm||C_{max}$, unless ETH fails.
\end{theorem}

\begin{theorem}
\label{th:length of the input1} For any $\delta
> 0$, there is no $2^{O(m^{1/2-\delta}
\sqrt{|I|})}$ time exact algorithm for $Pm||C_{max}$,  unless ETH
fails. Here $|I|$ is the length of the input.
\end{theorem}
\vspace{-2mm}

We also prove the traditional dynamic programming algorithm for the
scheduling problem runs in $2^{O(\sqrt{m|I|\log m}+ m\log
|I|)}$ time, and is thus essentially the best exact algorithm in
terms of running time. An overview about the known and new results
for $P||C_{max}$ is in the Table 1.
\vspace{-6mm}
\begin{table}[!hbp]
\begin{center}
\caption{Lower and upper bounds on the running
time}
\begin{tabular}{|c|c|c|}
\hline Algorithms & Upper bounds & Lower bounds\\
\hline Approximation scheme
&$2^{O(1/\epsilon^2\log^3(1/\epsilon))}+n^{O(1)}$&
$2^{O((1/{\epsilon})^{1-\delta})}+n^{O(1)}$\\
\hline Approximation scheme& $(1/\epsilon)^{O(m)}+O(n)$ &
$(1/\epsilon)^{O(m^{1-\delta})}+n^{O(1)}$\\
\hline Exact algorithm & $2^{O(\sqrt{m|I|\log m}+ m\log
|I|)}$ &$2^{O(m^{1/2-\delta}\sqrt{|I|})}$\\
 \hline
 Exact algorithm & $2^{O(n)}$ &$2^{O(n^{1-\delta})}$ ($O(n)$ jobs and processing times)\\
 \hline
\end{tabular}
\label{table:running time}
\end{center}
\end{table}
\vspace{-6mm}

Our results imply that the existing exact and approximation algorithms for the scheduling problem
are essentially the best possible, except a minor gap for the PTAS of $P||C_{max}$ that runs in $2^{O(1/\epsilon^2\log^3(1/\epsilon))}+n^{O(1)}$ time.
Given our results, it seems that a $2^{1/\epsilon\log^{O(1)}(1/\epsilon)}+n^{O(1)}$ time PTAS is possible, and it is indeed the case under certain
conditions. Jansen and Robenek
\cite{jansen chri} give a $2^{O(1/\epsilon
\log^2(1/\epsilon))} + n^{O(1)}$ time PTAS under a certain conjecture, and Chen et al. \cite{cardinality} provide a $2^{O(1/\epsilon\log^2(1/\epsilon))}+O(n)$
time PTAS if every machine can accept only a constant number of jobs (the lower bound also holds true if we restrict that each machine can accept at most $c$ jobs
for $c\ge 4$).

Briefly speaking, Theorem~\ref{th:id-sche-linear} and Theorem~\ref{th:id-sche-number}
rely on a nearly linear reduction, which reduces the 3SAT problem
with $n$ clauses and at most $3n$ variables to the scheduling
problem whose (optimal) makespan is bounded by $O(n^{1+\delta})$ for
any $\delta>0$. In contrast, the traditional reduction constructs a scheduling problem whose
makespan is bounded by $O(n^{16})$~\cite{garey}, and thus yields a
lower bound of $2^{(1/\epsilon)^{1/16}}$ assuming ETH. Theorem~\ref{th:sche-machine parameter} and Theorem~\ref{th:length
of the input1} rely on a different reduction, which reduces the 3SAT
problem with $O(n)$ variables and clauses to the scheduling problem
on $m$ machines whose makespan is bounded by
$2^{O(n/m\log^{O(1)}m)}$. The traditional reduction~\cite{garey}, however, is not able to characterize the
dependency on the number of machines. We remark that the framework of our reductions can also lead
to other interesting results, for example, to prove the APX-hardness~\cite{note} of the low rank scheduling problem
mentioned in~\cite{lowrank}.

\vspace{-5mm}
\section{Scheduling on Arbitrary Number of Machines}
\vspace{-2mm}

\begin{theorem}
\label{th:sche-makespan} Assuming ETH, there is no
$2^{O(K^{1-\delta})}|I_{sche}|^{O(1)}$ time algorithm which
determines whether there is a feasible schedule of makespan no more
than $K$ for any $\delta>0$.
\end{theorem}
\vspace{-3mm}
We prove the above theorem in this section, and then Theorem~\ref{th:id-sche-linear} follows
directly. To see why, suppose Theorem~\ref{th:id-sche-linear} fails,
then for some $\delta_0>0$ there exists a
$2^{O((1/{\epsilon})^{1-\delta_0})}\,\, n^{O(1)}$ time PTAS. We use
this algorithm to test if there is a feasible schedule of makespan
no more than $K$ for the scheduling problem by taking
$\epsilon=1/(K+1)$. If there exists a feasible schedule of makespan
no more than $K$, then the PTAS returns a solution with makespan no
more than $K(1+\epsilon)<K+1$. Otherwise, the makespan of the
optimal solution is larger than or equal to $K+1$, and the PTAS thus
returns a solution at least $K+1$. In a word, the PTAS determines
whether there is a feasible schedule of makespan no more than $K$ in
$2^{O((K+1)^{1-\delta_0})}\,\, |I|^{O(1)}$ time, which is a
contradiction to Theorem~\ref{th:sche-makespan}.

To prove Theorem~\ref{th:sche-makespan}, we start with a modified
version of the 3SAT problem, say, 3SAT' problem, in which the set of
clauses could be divided into sets $C_1$ and $C_2$ such that
\vspace{-2mm}
\begin{itemize}
\setlength\itemsep{-1mm}
\item In $C_1$, every clause contains three variables, and every variable appears once
\item In $C_2$, every clause is of the form $(z_i\vee\neg z_k)$, and every positive (negative) literal appears once
\item If a 3SAT' instance is satisfiable, every clause of $C_2$ is satisfied by exactly one literal
\end{itemize}
\vspace{-2mm}

There is a reduction from 3SAT (with $m$ clauses) to 3SAT' via
Tovey's method~\cite{tovey} which only increases the number of
clauses and variables by $O(m)$, and thus ensures the following
lemma (The reader may refer
to Lemma 1 in the full version for details).
\vspace{-2.5mm}
\begin{lemma}
\label{le:3sat'}Assuming ETH, there exists some $s>0$ such that
there is no $2^{sn}$ time algorithm for the 3SAT' problem with $n$
variables.
\end{lemma}
\vspace{-2.5mm}
Given a 3SAT' instance $I_{sat}$ with $n$ variables, we construct a
scheduling instance with $O(n/\delta)$ jobs and $O(n/\delta)$
machines such that it admits a feasible solution with makespan no
more than $K=O(2^{3/\delta}n^{1+\delta})$ if and only if $I_{sat}$
is satisfiable. This would be enough to prove
Theorem~\ref{th:sche-makespan}. To see why, suppose the theorem
fails, then an exact algorithm of running time
$2^{O(K^{1-\delta_0})}|I_{sche}|^{O(1)}$ exists for some
$\delta_0>0$. We take $\delta=\delta_0$ in the reduction, and we can
determine in $2^{O(n^{1-\delta_0^2})}|I|^{O(1)}=2^{o(n)}$ time
whether the constructed scheduling instance admits a schedule of
makespan $K$, and thus determine whether the given 3SAT' instance is
satisfiable, which is a contradiction.

For simplicity, all the subsequent proofs in this section take
$\delta=1/2$ and are thus simplified versions of that in the
full version, nevertheless, the main idea is
similar.
\vspace{-4mm}
\subsection{Overview of the reduction}
\vspace{-2mm}
\label{sec:overview}
We construct 5 kinds of jobs, among them variable jobs
and clause jobs correspond to the variables and clauses respectively, and are
the 'key jobs'. Other jobs serve as 'assistant jobs', including the huge jobs, dummy jobs
and truth-assignment jobs.

Recall that every positive (negative) literal
appears at most twice, while every clause of $C_1$ contains three literals. For each positive (negative) literal, say,
$z_i$ (or $\neg z_i$), two pairs of jobs $v_{i,1}^{\gamma}$ and
$v_{i,2}^{\gamma}$ ($v_{i,3}^{\gamma}$ and $v_{i,4}^{\gamma}$) are
constructed where $\gamma\in\{T,F\}$. For each clause of $C_1$, say,
$c_j$, one job $u_j^T$ and two copies of job $u_j^F$ are
constructed.

We construct huge jobs to create gaps. Precisely speaking, every huge job has a processing time larger than $1/2K$,
and thus one huge job occupies one machine, leaving a gap if all the jobs on this machine
should add up to $K$. Depending on their sizes, all the gaps form a staircase structure, and
the scheduling problem becomes to determine whether the remaining jobs fit into all the gaps.

Suppose we can design 4 kinds of gaps (huge jobs) satisfying the following condition:

\vskip 0.5mm\noindent\textbf{[1.]} Variable-assignment gaps. To fill
up these gaps, for any $i$ either $v_{i,1}^F$, $v_{i,2}^F$, $v_{i,3}^T$,
$v_{i,4}^T$, or $v_{i,1}^T$,
$v_{i,2}^T$, $v_{i,3}^F$, $v_{i,4}^F$ are used.

\vskip 0.5mm\noindent\textbf{[2.]} Variable-clause gaps. If the
positive (or negative) literal $z_i$ (or $\neg z_i$) is in $c_j\in
C_1$, then a variable-clause gap is created so that it could only be
filled up by $u_j$ and $v_{i,1}$ (or $v_{i,3}$). Furthermore, for the
superscripts of $u_j$ and $v_{i,1}$ (or $v_{i,3}$), the gap enforces
that only three combinations are valid: (T,T), (F,F) and (F,T).

\vskip 0.5mm\noindent\textbf{[3.]} Variable-agent and agent-agent
gaps. To fill up these gaps, for any $(z_i\vee\neg z_k)\in C_2$
either $v_{i,2}^T$ and $v_{k,4}^F$, or
$v_{i,2}^F$ and $v_{k,4}^T$ are used.

\vskip 0.5mm\noindent\textbf{[4.]} Variable-dummy gaps. Recall that 8 variable jobs
are constructed for a variable and only 7 of them are used (either
$v_{i,1}$ or $v_{i,3}$ is left), the remaining one will be used to
fill these gaps.

It is not difficult to verify that if every gap is filled up, $I_{sat}$ is satisfiable. To
see why, if $v_{i,1}^F$, $v_{i,2}^F$, $v_{i,3}^T$,
$v_{i,4}^T$ are used in the variable-assignment gaps, then we let variable $z_i$ be true, otherwise
we let it be false. For any clause of $C_1$, say, $c_j$, there is one $u_j^T$ and it must
be scheduled with a true variable job, say, $v_{i,1}^T$ (or $v_{i,3}^T$). If $v_{i,1}^T$ is scheduled
with $u_j^T$, then the positive literal $z_i$ is in $c_j$. Meanwhile the variable $z_i$ is true since
otherwise $v_{i,1}^T$ are used to fill variable-assignment gaps. Thus $c_j$ is satisfied. Similar argument
shows that $c_j$ is also satisfied if $v_{i,3}^T$ is with $u_j^T$. For any clause of $C_2$, say, $(z_i\vee\neg z_k)$,
if $v_{i,2}^T$ and $v_{k,4}^F$ ($v_{i,2}^F$ and $v_{k,4}^T$) are used to fill up the third type of gaps, then
it is easy to verify that variables $z_i$ and $z_k$ are both true (false), implying that $(z_i\vee\neg z_k)$ is satisfied (by exactly
one literal).

\noindent\textbf{Technical Part. }
The difficult part of the reduction is that, how can we design the size of a gap so that it is
filled up by two specific jobs. Roughly speaking, to ensure that a gap is filled up by $\alpha_i$ and
$\beta_j$ rather than $\alpha_{i'}$ and $\beta_{j'}$, a straightforward way is to create 'gaps' between $\alpha_i$
and $\beta_j$. If we let $r=\Theta(n)$, and let $\alpha_i=ir$, $\beta_j=j<r$, then a gap of $ir+j$ has to be
filled up by $\alpha_i$ and $\beta_j$, but then $ir+j=\Theta(n^2)$, which is not favorable.

We notice that, if the indices $i$
and $j$ are not 'free', but satisfies $|i-j|\le h$ for some $h$, then a gap of $\Theta(hr)$ suffices. To
get an intuition, suppose we want to ensure that $\alpha_i$ is always scheduled with $\beta_j=\beta_{i+h}$ for $i=1,\cdots,n$,
then we let $\beta_j=j-h$ (here $j=h+1,\cdots,h+n$), $\alpha_i=hr+i$, and create $n$ gaps of size $hr+2i$ for $i=1,\cdots,n$. Then
$hr+2$ has to be filled up $\beta_{1+h}$ and $\alpha_{1}$. Given that $\beta_{h+1}$ and $\alpha_{1}$ are scheduled, $hr+4$ has
to be filled up by $\beta_{h+2}$ and $\alpha_{2}$, and so on.

How can we establish relationship between indices?
Recall that every variable appears once in $C_1$, we re-index
variables and clauses of $C_1$ so that clause $c_i\in C_1$ contains
three variables $z_i$, $z_{i+1}$ and $z_{i+2}$ where $i\in
R=\{1,4,7,\cdots, n\}$. In this way variable-claus gaps could be constructed using the above idea.

For variable-assignment gaps, we create truth assignment jobs $a_i^{\gamma}$, $b_i^{\gamma}$, $c_i^{\gamma}$
and $d_i^{\gamma}$ as assistant jobs, and create 4 gaps so that they admit $(v_{i,1},a_i,c_i)$, $(v_{i,2},b_i,d_i)$,
$(v_{i,3},a_i,d_i)$, $(v_{i,4},b_i,c_i)$ respectively, and the three jobs for each gap are either all true or
all false. It is not difficult to verify that in this way either $v_{i,1}^F$, $v_{i,2}^F$, $v_{i,3}^T$,
$v_{i,4}^T$, or $v_{i,1}^T$,
$v_{i,2}^T$, $v_{i,3}^F$, $v_{i,4}^F$ are used, and furthermore, the indices of the three jobs are the same, and
again we may use the above idea.

Consider any clause of $C_2$, say, $(z_i\vee\neg z_k)$. Since variables have been re-indexed, indices $i$
and $k$ are arbitrary and it is possible that $|i-k|=O(n)$. To handle this, we try to map indices $i$ and
$k$ to $i'$ and $k'$ respectively, such that $|i'-k'|\le O(\sqrt{n})$, and then gaps of $O(n^{3/2})$ suffice. Precisely,
for $v_{i,2}$ ($v_{k,4}$), we construct a pair of agent jobs, namely $\eta_{i,+}^{\gamma}$ (or
$\eta_{i,-}^{\gamma}$) where $\gamma=\{T,F\}$. We create one variable-agent gap which could
only be filled up by $v_{i,2}$ and its agent $\eta_{i,+}$, and they
should be one true and one false (i.e., their superscripts are $T$
and $F$). Similarly another variable-agent gap is created which
could only be filled up by $v_{k,4}$ and $\eta_{k,-}$ that are one
true and one false. We further create an agent-agent gap which could
only be filled up by $\eta_{i,+}$ and $\eta_{k,-}$ that are one true
and one false. Combining the three gaps, we can conclude that the
$v_{i,2}$ and $v_{k,4}$ used in these gaps are one true and one
false. Indeed, such a method changes the designing of a gap that enforces
$v_{i,2}$ and $v_{k,4}$ are together to the designing of a gap that enforces
their agent jobs are together. The 'agent index' $i'$ of $i$ is implicitly implied in the
processing time of $\eta_{i,+}$, and is defined through the functions $f$ and $g$,
as is shown in the following.
\vspace{-2.5mm}
\paragraph{Defining functions $f$ and $g$.}
\vspace{-2.5mm}
\label{subsection:modi-sat-1} Given a 3SAT' instance $I_{sat}$ with
$n$ variables, we know that $|C_1|=n/3$ and $|C_2|=n$. We may
further assume that $n$ is sufficiently large (i.e., $n\ge 2^{26}$)
and $\sqrt{n}$ is an integer.

Recall that there are $n$ clauses in $C_2$ and every positive
(negative) literal appears once in them. We partition clauses of
$C_2$ equally into $\sqrt{n}$ groups. Let $S_n=\{1,2,\cdots,n\}$. We
define the function $f: S_n\rightarrow S_{\sqrt{n}}$ such that the
positive literal $z_i$ is in group $f(i)$, and we define the
function $\bar{f}: S_n\rightarrow S_{\sqrt{n}}$ such that the
negative literal $\neg z_i$ is in group $\bar{f}(i)$.

In each group, say, group $i$, there are $\sqrt{n}$ different
positive literals. Let their indices be $i_1<
i_2<\cdots<i_{\sqrt{n}}$, then we define $g: S_n\rightarrow
S_{\sqrt{n}}$ such that $g(i_k)=k$. Similarly the indices of
negative literals could be listed as $\bar{i}_1<
\bar{i}_2<\cdots<\bar{i}_{\sqrt{n}}$ and we define $\bar{g}:
S_n\rightarrow S_{\sqrt{n}}$ such that $\bar{g}(\bar{i}_k)=k$.

Our definition of $g$ and $\bar{g}$ implies the following lemma.
\vspace{-2mm}
\begin{lemma}
\label{le:sympl g and bar g} For any $i,i'\in S_n$ and $i<i'$, if
$f(i)=f(i')$, then $g(i)<g(i')$. Similarly if
$\bar{f}(i)=\bar{f}(i')$, then $\bar{g}(i)<\bar{g}(i')$.
\end{lemma}
\vspace{-6mm}
\subsection{Construction of the Scheduling Instance}
\vspace{-1mm}
Given $I_{sat}$, we construct an instance of scheduling problem with
$30n$ jobs and $9n$ machines, and prove that $I_{sat}$ is
satisfiable if and only if there exists a feasible solution for the
constructed scheduling instance with makespan no more than
$K=10^5r$, where $r=2^{15}n^{3/2}$. Throughout this section we set
$x=4\sqrt{n}$ and use $s(j)$ to denote the processing time of job
$j$. $\gamma\in\{T,F\}$.

$20n$ jobs are constructed for variables, among them there are $8n$
variable jobs, $4n$ agent jobs and $8n$ truth assignment jobs. $n$
jobs are constructed for clauses of $C_1$. $9n$ huge jobs are
constructed to create gaps.

Variable jobs: $v_{i,1}^{\gamma}$ and $v_{i,2}^{\gamma}$ are
constructed for $z_i$, $v_{i,3}^{\gamma}$ and $v_{i,4}^{\gamma}$ are
for $\neg z_i$.
$$\setlength{\abovedisplayskip}{2pt}
\setlength{\belowdisplayskip}{2pt}
s(v_{i,k}^T)=r+2^{9}(f(i)x^2+i)+2^{8}+k,\quad
k=1,2.$$
$$\setlength{\abovedisplayskip}{2pt}
\setlength{\belowdisplayskip}{2pt} s(v_{i,k}^T)=r+2^{9}(\bar{f}(i)x^2+i)+2^{8}+k,\quad
k=3,4.$$
$$\setlength{\abovedisplayskip}{2pt}
\setlength{\belowdisplayskip}{2pt} s(v_{i,k}^F)=s(v_{i,k}^T)+2r,\quad k=1,2,3,4$$
Agent jobs: $\eta_{i,+}^{\gamma}$ for $z_i$ and $\eta_{i,-}^{\gamma}$ for $\neg z_i$.
$$\setlength{\abovedisplayskip}{2pt}
\setlength{\belowdisplayskip}{2pt} s(\eta_{i,+}^T)=r+2^{9}(f(i)x^2+g(i))+2^{7}+8, $$
$$\setlength{\abovedisplayskip}{2pt}
\setlength{\belowdisplayskip}{2pt} s(\eta_{i,-}^T)=r+2^{9}(\bar{f}(i)x^2+\bar{g}(i)x)+2^{7}+16.$$
$$\setlength{\abovedisplayskip}{2pt}
\setlength{\belowdisplayskip}{2pt} s(\eta_{i,\sigma}^F)=s(\eta_{i,\sigma}^T)+2r, \quad\sigma=+,-$$

Truth assignment jobs: $a_i^{\gamma}$, $b_i^{\gamma}$,
$c_i^{\gamma}$ and $d_i^{\gamma}$.
$$\setlength{\abovedisplayskip}{2pt}
\setlength{\belowdisplayskip}{2pt} s(a_i^F)=11r+(2^7i+8), \quad s(b_i^F)=11r+(2^7i+32),$$
$$\setlength{\abovedisplayskip}{2pt}
\setlength{\belowdisplayskip}{2pt} s(c_i^F)=101r+(2^7i+16), \quad s(d_i^F)=101r+(2^7i+64).$$
$$\setlength{\abovedisplayskip}{2pt}
\setlength{\belowdisplayskip}{2pt} s(k_i^{T})=s(k_i^F)+r,\quad k=a,b,c,d.$$

Clause jobs: 3 clause jobs are constructed for every $c_{j}\in C_1$
where $j\in R$, with one $u_{j}^T$ and two copies of $u_{j}^F$:
$s(u_{j}^T)=10004r+2^{11}j, s(u_{j}^F)=10002r+2^{11}j.$

Dummy jobs: $n+n/3$ jobs of $1000r$, and $n-n/3$
jobs of $1002r$.

Let
$A$, $B$, $C$, $D$ be the set of $a_i^{\gamma}$, $b_{i}^{\gamma}$,
$c_i^{\gamma}$ and $d_i^{\gamma}$ respectively. Sometimes we may
drop the superscript for simplicity, e.g., we use $a_i$ to represent
$a_i^T$ or $a_i^F$.

We construct huge jobs. There are four kinds of huge jobs
corresponding to the four kinds of gaps we mention before.

Two huge jobs (variable-agent jobs) $\theta_{\eta,i,+}$ and
$\theta_{\eta,i,-}$ are constructed for each variable $z_i$:
\vspace{-1.5mm}
\begin{eqnarray*}
\setlength{\abovedisplayskip}{1pt}
\setlength{\belowdisplayskip}{1pt}
s(\theta_{\eta,i,+})&=&10^5r-4r-2^{9}[2f(i)x^2+g(i)+i]-(2^{8}+2^{7}+10)\\
s(\theta_{\eta,i,-})&=&10^5r-4r-2^{9}[2\bar{f}(i)x^2+\bar{g}(i)x+i]-(2^{8}+2^{7}+20)
\end{eqnarray*}
\vspace{-3mm}
One huge job (agent-agent job) $\theta_{i,k,C_2}$ is constructed for
$(z_i\vee \neg z_k)\in C_2$:
$$s(\theta_{i,k,C_2})=10^5r-4r-2^{9}[f(i)x^2+\bar{f}(k)x^2+\bar{g}(k)x+g(i)]+2^{8}+24].$$

\vspace{-1mm}
Notice that $f(i)=\bar{f}(k)$ according to our definition of $f$ and
$\bar{f}$.

Three huge jobs (variable-clause jobs) are constructed for each
$c_{j}\in C_1$ ($j\in R$), one for each literal: for $i=j,j+1,j+2$,
if $z_i\in c_j$, we construct $\theta_{j,i,+,C_1}$, otherwise $\neg
z_i\in c_j$, and we construct $\theta_{j,i,-,C_1}$.
$$\setlength{\abovedisplayskip}{1pt}
\setlength{\belowdisplayskip}{1pt} s(\theta_{j,i,+,C_1})=10^5r-11005r-(2^{9}f(i)x^{2}+2^{11}j+2^{9}i+2^{8}+1),$$
$$\setlength{\abovedisplayskip}{1pt}
\setlength{\belowdisplayskip}{1pt} s(\theta_{j,i,-,C_1})=10^5r-11005r-(2^{9}\bar{f}(i)x^{2}+2^{11}j+2^{9}i+2^{8}+3).$$

One huge job (variable-dummy job) is constructed for each variable.
Notice that each variable appears exactly three times in clauses, if
$z_i$ appears twice while $\neg z_i$ appears once, we construct
$\theta_{i,-}$. Otherwise, we construct $\theta_{i,+}$ instead.
$$\setlength{\abovedisplayskip}{1pt}
\setlength{\belowdisplayskip}{1pt} s(\theta_{i,+})=10^5r-1003r-(2^{9}f(i)x^2+2^{9}i+2^{8}+1),$$
$$\setlength{\abovedisplayskip}{1pt}
\setlength{\belowdisplayskip}{1pt} s(\theta_{i,-})=10^5r-1003r-(2^{9}\bar{f}(i)x^2+2^{9}i+2^{8}+3).$$
\vspace{-1mm}
Thus, for each clause $c_j$ ($j\in R$) and $i=j,j+1,j+2$, either
$\theta_{i,+}$ and $\theta_{j,i,-,C_1}$ exist, or $\theta_{i,-}$ and
$\theta_{j,i,+,C_1}$ exist.

Four huge jobs (variable-assignment jobs) are constructed for each
variable $z_i$, namely $\theta_{i,a,c}$, $\theta_{i,b,d}$,
$\theta_{i,a,d}$ and $\theta_{i,b,c}$:
$$\setlength{\abovedisplayskip}{1pt}
\setlength{\belowdisplayskip}{1pt} s(\theta_{i,a,c})=10^5r-115r-2^{9}(f(i)x^{2}+i)-(2^{8}+2^8i+25),$$
$$\setlength{\abovedisplayskip}{1pt}
\setlength{\belowdisplayskip}{1pt} s(\theta_{i,b,d})=10^5r-115r-2^{9}(f(i)x^{2}+i)-(2^{8}+2^8i+98),$$
$$\setlength{\abovedisplayskip}{1pt}
\setlength{\belowdisplayskip}{1pt} s(\theta_{i,a,d})=10^5r-115r-2^{9}(\bar{f}(i)x^{2}+i)-(2^{8}+2^8i+75),$$
$$\setlength{\abovedisplayskip}{1pt}
\setlength{\belowdisplayskip}{1pt} s(\theta_{i,b,c})=10^5r-115r-2^{9}(\bar{f}(i)x^{2}+i)-(2^{8}+2^8i+52).$$

It is not difficult to verify that the total processing time of all
the jobs is $9n\cdot 10^5r$. Furthermore, if the given 3SAT'
instance $I_{sat}$ is satisfiable, then the constructed scheduling
instance $I_{sche}$ admits a feasible schedule whose makespan is
$10^5r$ (the reader may refer to Appendix~\ref{appendix-subsection:
sat to schedule} for details).

\vspace{-4mm}
\subsection{Scheduling to 3SAT}
\vspace{-1.5mm}
We prove that if there is a schedule whose makespan is no more than
$10^5r$ (which implies that the load of each machine is exactly
$10^5r$), then $I_{sat}$ is satisfiable. To achieve this, we only need
to show that to fill up the gaps (created by huge jobs), the key jobs have to be scheduled
in the way as we mention in Subsection~\ref{sec:overview}.

Recall that
we define the processing time of a job in the form of a polynomial,
which could be partitioned into four terms, the $r$-term,
$x^2$-term, $x$-term and constant term (the summation of all terms
without $r$ or $x$). For simplicity, the sum of the $x$-term and constant term is
called small-$x^2$-term, and the sum of $x^2$-term, $x$-term and
constant term is called small-$r$-term. Since $x=4\sqrt{n}$ and
$r=2^{15}n^{3/2}$, there are gaps between terms.

\vspace{-3mm}
\begin{lemma}
\label{le:sympl terms} The small-$r$-term and small-$x^2$-term of a
huge job are negative with their absolute values bounded by $1/2r$
and $2^9\cdot 3/4x^2$ respectively. The small-$r$-term of any other
job is positive and bounded by $1/4r$. The small-$x^2$-term of a
variable or agent job is positive and bounded by $2^9\cdot 3/8x^2$.
\end{lemma}
\vspace{-3mm}

Notice that the input of the scheduling instance is a set of integers (processing times),
the above lemma allows us to determine the symbol of a job through its processing time.
Furthermore, by considering the $r$-terms and the residuals of dividing
each job by $2^7$, the following observation is true through a counting argument (the reader may refer to
Lemma~\ref{le:sche-3sat'-step 1} in the
Appendix~\ref{appendix:scheduling}).

\noindent\textbf{Observation.}
\vskip 1mm\noindent\textbf{[1.]}{A variable-agent gap is filled up with a variable job and an agent job.}
\vskip 1mm\noindent\textbf{[2.]}{An agent-agent gap is filled up with two agent jobs.}
\vskip 1mm\noindent\textbf{[3.]}{A variable-clause gap is filled up with a clause job, a variable job and a dummy job.}
\vskip 1mm\noindent\textbf{[4.]}{A variable-dummy gap is filled up with a variable job and a dummy job.}
\vskip 1mm\noindent\textbf{[5.]}{A variable-assignment gap is filled up with a variable job and two truth-assignment jobs,
one in $A\cup B$, the other in $C\cup D$.}

Combining the above observation with Lemmas~\ref{le:sympl terms}, we
get the following lemma.

\vspace{-2.5mm}
\begin{lemma}
\label{le:sympl r and delta equal} For jobs on each machine, their
$r$-terms add up to $10^5r$, $x^{2}$-terms and small-$x^2$-terms add
up to $0$.
\end{lemma}
\vspace{-2.5mm}

Consider the $x^2$-terms of gaps. An agent-agent gap or
variable-agent gap is called a regular gap, since their $x^2$-terms
are $2^9\cdot 2\zeta x^2$ where $1\le \zeta\le \sqrt{n}$. Other gaps
are called singular gaps with the $x^2$-terms being $2^9\zeta x^2$.
A singular gap is called well-canceled, if it is filled up by other
jobs whose $x^2$-terms are $2^9\zeta x^2$ and $0$. A regular gap is
called well-canceled, if it is filled up by two jobs whose
$x^2$-terms are both $2^9\zeta x^2$.

\vspace{-2.5mm}
\begin{lemma}
\label{le:sympl singular well-cancel} Every singular gap is
well-canceled, and every regular gap is well-canceled.
\end{lemma}
\begin{prf}
We briefly argue why it is the case. The first part follows directly
from Lemma~\ref{le:sympl r and delta equal}. We show the second
part.

Consider the regular gap with the term $2^9\cdot 2x^2$. Since it is
filled up by two variable or agent jobs (due to Observation),
whose $x^2$-term is at least $2^9x^2$, thus it is obviously
well-canceled.

According to the construction of $f$ and $\bar{f}$, there are
$\sqrt{n}$ indices such that $f(i)=1$ and $\sqrt{n}$ indices such
that $\bar{f}(i)=1$, thus in all there are $12\sqrt{n}$ variable and
agent jobs with the term $2^9x^2$. There are $\sqrt{n}$
variable-clause gaps, $\sqrt{n}$ variable-dummy gaps and $4\sqrt{n}$
variable-assignment gaps with the term $2^9x^2$, meanwhile there are
$2\sqrt{n}$ variable-agent gaps and $\sqrt{n}$ agent-agent gaps with
the term $2^9\cdot 2x^2$. All of these gaps are well-canceled,
implying that all the variable and agent jobs with $2^9x^2$ are used
to fill up these gaps. Thus to fill up a regular gap with $2^9\cdot
4x^2$, we have to use variable or agent jobs with $2^9\cdot 2x^2$,
implying that this regular gap is also well-canceled. Iteratively
applying the above arguments, every regular gap is well-canceled.
\end{prf}

A huge job (gap) is called satisfied, if the indices of other jobs
on the same machine with it coincide with its index. For example,
the variable-clause job $\theta_{j,i,-,C_1}$ is satisfied if it is
on the same machine with the variable job $v_{i,k}$ and clause job
$u_j$ where $k\in\{1,2,3,4\}$, and $\theta_{i,a,c}$ is satisfied if
it is with $v_{i,k}$, $a_i$ and $c_i$.

\vspace{-2.5mm}
\begin{lemma}
\label{le:sympl satisfy} Every huge job (gap) is satisfied.
\end{lemma}
\begin{prf}
We give the sketch of proof. It is easy to see that every
variable-dummy job is satisfied. According to the definition of $f$
and $g$ ($f'$ and $g'$), an index $i$ is determined uniquely by the
pair $(f(i),g(i))$ (or $(\bar{f}(i),\bar{g}(i))$). Combining this
fact with Lemma~\ref{le:sympl singular well-cancel}, it is not
difficult to verify that every agent-agent job $\theta_{i,k,C_2}$ is
scheduled with $\eta_{i,\sigma}$ and $\eta_{k,\sigma'}$ where
$\sigma,\sigma'\in\{+,-\}$, and is thus satisfied.

Consider the variable-agent job $\theta_{\eta,1,+}$. According to
the observation and Lemma~\ref{le:sympl singular
well-cancel}, the gap of $4r+2^9(2f(1)x^2+g(1)+1)+2^8+2^7+10$ should
be filled up by a variable job $v_{i',k}$ and an agent job
$\eta_{i'',\sigma}$ where $k\in\{1,2,3,4\}$ and $\sigma\in\{+,-\}$,
such that $f(i')=f(i'')=1$. Simple calculations show that
$g(i'')+i'=g(1)+1$. Since $i,i'\ge 1$, Lemma~\ref{le:sympl g and bar
g} implies that $i'=i''=1$, and $\theta_{\eta,1,+}$ is thus
satisfied. Similarly we can prove that $\theta_{\eta,1,-}$ is
satisfied.

Using similar arguments, it is not difficult to verify that the
three variable-clause job $\theta_{1,i,\sigma_i,C_1}$
($\sigma_i\in\{+,-\}$ for $i=1,2,3$) and the four
variable-assignment jobs $\theta_{1,a,c}$, $\theta_{1,b,d}$,
$\theta_{1,a,d}$, $\theta_{1,b,c}$ are satisfied. We call $v_{i,k}$
($k\in\{1,2,3,4\}$) and $\eta_{i,\sigma}$ ($\sigma\in\{+,-\}$) as
jobs of index-level $i$, then all the jobs of index-level 1 are used
to fill up the the previous mentioned gaps so that when we consider
$\theta_{\eta,2,+}$, it should be scheduled together with $v_{i',k}$
and $\eta_{i'',\sigma}$ with $i',i''\ge 2$, and we can carry on the
previous arguments.
\end{prf}
\vspace{-2.5mm}
The reader may refer to Lemma~\ref{le:each satisfy} of
Appendix~\ref{appendix:scheduling} for details of the above proof.
With the above lemma, it is not difficult to further verify (due to
the residuals of each job divided by $2^7$) that variable jobs are scheduled
according to the following table. (Recall that for every $j\in R$ and
$i=j,j+1,j+2$, either $\theta_{j,i,+,C_1}$ and $\theta_{i,-}$ exist,
or $\theta_{j,i,-,C_1}$ and $\theta_{i,+}$ exist.)

\vspace{-2mm}
\begin{table}[!hbp]
\begin{center}
\begin{tabular}{|c|c||c|c||c|c||c|c||c|c||}
\hline $\theta_{i,a,c}$ & $v_{i,1}$ & $\theta_{i,a,d}$ & $v_{i,3}$ & $\theta_{j,i,+,C_1}$ & $v_{i,1}$ & $\theta_{i,+}$ & $v_{i,1}$ & $\theta_{\eta,i,+}$ & $v_{i,2}$ \\
\hline $\theta_{i,b,d}$ & $v_{i,2}$ & $\theta_{i,b,c}$ & $v_{i,4}$ & $\theta_{j,i,-,C_1}$ & $v_{i,3}$ & $\theta_{i,-}$ & $v_{i,3}$ & $\theta_{\eta,i,-}$ & $v_{i,4}$ \\
 \hline
\end{tabular}
\label{table:sympl-jobs}
\end{center}
\end{table}
\vspace{-6mm}

The previous discussion determines the indices of jobs on each
machine, and we can further determine their superscripts by considering the $r$-terms
of jobs.
\vspace{-2.5mm}
\begin{itemize}
\setlength\itemsep{-1.5mm}
\item The two jobs to fill up an agent-agent or variable-agent gap are one true and one false.
\item The three jobs to fill up a variable-assignment gap are either (T,T,T) or (F,F,F).
\item The clause job and variable job to fill up a variable-clause gap are
(T,T), (F,F) or (F,T).
\end{itemize}
\vspace{-2.5mm}

Now we can conclude that, to fill up all the gaps, the jobs scheduled satisfy the
4 conditions in Subsection~\ref{sec:overview}, and thus $I_{sat}$ is satisfiable.

\noindent\textbf{Remark. }
The processing time of an agent job should be defined in a proper
way so that we can determine from the gaps that a variable job is
scheduled with its corresponding agent job, and two specific agent
jobs are scheduled together, and this requires a processing time of
$O(n^{3/2})$. To reduce it to $O(n^{1+\delta})$ for $\delta>0$, we
create $1/\delta-1$ pairs of agent jobs (from layer-1 to
layer-$(1/\delta-1)$) for a variable. A variable job is scheduled
with its layer-$(1/\delta-1)$ agent job, its layer-$(1/\delta-1)$
agent job is with its rank-$(1/\delta-2)$ agent job, $\cdots$, its
layer-2 agent job is with its layer-1 agent job, and two specific
layer-1 agent jobs are scheduled together. The reader is referred to
Appendix~\ref{appendix:scheduling} for details.

\vspace{-5mm}
\section{Scheduling on $m$ Machines}
\vspace{-3mm}
\begin{theorem}
\label{th:sche-machine}Assuming ETH, there is no
$(1/\epsilon)^{o(\frac{m}{\log^2 m})}|I_{sche}|^{O(1)}$ time FPTAS
for $Pm||C_{max}$.
\end{theorem}
\vspace{-2mm}
We prove the above theorem in this section, and Theorem~\ref{th:sche-machine parameter}
follows since otherwise, there exists a
$(1/\epsilon)^{O(m^{1-\delta_0})} \,\, |I|^{O(1)}$ time FPTAS for
some $\delta_0>0$, and it runs in $(1/\epsilon)^{o(\frac{m}{\log^2
m})}|I_{sche}|^{O(1)}$ time, which is a contradiction.

To prove Theorem \ref{th:sche-machine}, given any 3SAT' instance
$I_{sat}$ with $n$ variables, we construct a scheduling instance
$I_{sche}$ such that it admits an optimal solution with makespan
$2^{O(\frac{n}{m}\log^2 m)}$ if and only if $I_{sat}$ is
satisfiable, then if the above theorem fails, we may apply the
$(1/\epsilon)^{o(\frac{m}{\log^2 m})}|I_{sche}|^{O(1)}$ time PTAS
for $I_{sche}$ by setting $1/\epsilon=2^{O(\frac{n}{m}\log^2 m)}+1$.
Simple calculations show that the optimal solution could be computed
in $2^{\delta_m n}$ time where $\delta_m$ goes to $0$ as $m$
increases, and thus the satisfiability of the given 3SAT' instance
could also be determined in $2^{\delta_m n}$ time, which is a contradiction.

For simplicity throughout the following we let the number of machines be $m+1$ (instead of $m$).

\noindent\textbf{Overview of the Reduction}

We first give a short explanation of the traditional reduction which
reduces the 3 dimensional matching problem (3DM) to $P2||C_{max}$~\cite{garey}. In the 3DM problem, there
are three disjoint sets of elements $W\cup X\cup Y$ with $|W|=|X|=|Y|=q$, and a set of matches
$T\subset W\times X\times Y$. The problem asks whether there exists a proper matching, i.e.,
a subset of $T$ in which every element appears once. Consider integers no more than $\alpha^{3q+1}-1$ (where
$\alpha$ is some parameter). Taking $\alpha$ as the base, there are $3q$ 'bits'\, in these integers, and the
traditional reduction allocates a bit to a distinct element. Suppose element $\lambda\in W\cup X\cup Y$ is allocated
with the $f(\lambda)$-th bit, then for every match $(w_i,x_j,y_k)\in T$, a job of processing time $\alpha^{f(w_i)}+\alpha^{f(x_j)}+\alpha^{f(y_k)}$
is constructed. These are called key jobs, and $\alpha$ is taken to be large enough so that when we add up all the key jobs,
there is no 'carry over'\, between bits (e.g., $\alpha=|T|+1$). Let $B=\sum_{i=1}^{3q}\alpha^i=(111\cdots11)_{\alpha}$. By
creating huge dummy jobs, the scheduling problem is equivalent to asking whether there is a subset of
key jobs adding up to $B$,
and it is easy to verify that these key jobs correspond to a proper matching.

There is a traditional reduction that reduces the 3SAT problem with $O(n)$ variables and clauses to the 3DM problem
with $O(n^2)$ elements~\cite{garey}, and thus yields a scheduling problem with makespan $2^{O(n^2)}$ when combined with
the above reduction. To reduce the size, we need to first give a linear reduction that reduces the 3SAT problem with
$O(n)$ variables and clauses to the 3DM problem
with $O(n)$ elements and matches. Indeed, this could be achieved by slightly generalizing the 3DM problem, as we will show in
the next part.

Now we try to allocate 'bits'\, to elements. Let $\alpha=m^{O(1)}$ be the base. Since we aim to construct a scheduling instance
with makespan bounded by $2^{n/m\log^{O(1)}m}$, there are only $n/m\log^{O(1)}m$ different bits. Recall that there are $O(n)$ elements.
Let $f$ be some allocation that maps the elements to $\{1,2,\cdots,|f|\}$ where $|f|\le n/m\log^{O(1)}m$, then $O(m)$ elements may
share the same bit.
Roughly speaking, we will again create a key job of processing time similar to $\alpha^{f(w_i)}+\alpha^{f(x_j)}+\alpha^{f(y_k)}$ for $(w_i,x_j,y_k)$.
By creating dummy jobs, the scheduling problem is equivalent to asking
whether there are $m$ disjoint subsets of key jobs with the
total processing time of jobs in each subset equal to
$\sum_{i=1}^{|f|}\alpha^i=(111\cdots11)_{\alpha}$.
Thus, in jobs of each subset every $\alpha^i$ term should appear once (meaning that jobs
in the same subset do not share the same $\alpha^i$ term) and the key jobs
on the $m$ machines will correspond to a proper matching.

The difficult part is from 3DM to scheduling. To make the above argument work, the allocation function
$f$ should satisfy a 'universal'\, property: for any proper matching of the 3DM instance (if it exists),
jobs corresponding to the matching could be divided into $m$ subsets such that in every subset jobs
do not share the same $\alpha^i$ term (i.e., the elements of the matches corresponding to jobs in each subset
are not mapped to the same bit). How can we design such a function $f$ without any knowledge of the matching? This
would be achieved by starting with a 3DM problem of a special structure, and using the idea of greedy coloring
in the underlying graph of the 3DM problem.

\noindent\textbf{From 3SAT to 3DM.}
Given a 3SAT' instance $I_{sat}$, by applying Tovey's
method~\cite{tovey} for a second time and a proper re-indexing of indices, we may further alter it
and then transform it into a 3DM instance $I_{3dm}$ with the following structure:
\vspace{-2.5mm}
\begin{itemize}
\setlength\itemsep{-1.5mm}
\item There are three disjoint sets of elements $W=\{w_i,\bar{w}_i|i=1,\cdots,3n\}$,
$X=\{s_j,a_i|j=1,\cdots,|C_1|,i=1,\cdots,3n\}$ and
$Y=\{b_i|i=1,\cdots,3n\}$
\item There are three sets of matches $T_1=\{(w_i),(\bar{w}_i)|i=1,\cdots,3n\}$,
$T_2\subseteq\{(w_i,s_j),(\bar{w}_i,s_j)|w_i\in W,s_j\in X\}$,
$T_3=\{(w_i,a_i,b_i),(\bar{w}_i,a_i,b_{\zeta{(i)}})|i=1,\cdots,3n\}$
where $\zeta$ is defined as $\zeta(3k+1)=3k+2$, $\zeta(3k+2)=3k+3$
and $\zeta(3k+3)=3k+1$ for $k=1,\cdots,n$
\item Either $w_i$ or $\bar{w}_i$ appears in $T_2$, and appears once. Every $s_j$ appears at most three times in
$T_2$.
\end{itemize}
\vspace{-2.5mm}
We remark that the above 3DM problem (denoted as 3DM') is actually a slight generalization of the traditional 3DM problem by
allowing one-element matches like $(w_i)$ and two-element matches like $(w_i,s_j)$, and as a consequence $|W|\ge |X|\ge |Y|$. Notice
that in the 3DM' problem, $T_1$ and $T_3$ are fixed. The 3DM' problem
also asks for the existence of a proper matching (a subset of $T_1\cup T_2\cup T_3$ where every element appears once). There is a linear reduction from 3SAT' to 3DM', implying the existence of
some $s'$ such that there is no $2^{s'n}$ time algorithm for the 3DM' problem under ETH. The reader may refer to Appendix~\ref{appendix:I_3dm} for details.

\vspace{-5mm}
\paragraph{Defining the Function $f$ Based on Partitioning Matches.} By
introducing dummy elements and matches we may assume that $n=qm$ for some
integer $q$. We divide the set $W$ equally into $m$ subsets, with
$W_k=\{w_i,\bar{w}_i|3kq+1\le i\le 3kq+3q\}$ for $0\le k\le m-1$. Given a proper
matching, every element appears once, thus we can always divide the matching into
$m$ subsets so that matches containing $w_i, \bar{w}_i\in W_k$ are in the $k$-th subset. We design
$f$ such that elements appear in the same subset are allocated with distinct bits.

We construct a bipartite graph $G=(V^w\cup V^s,E)$ in the
following way. There are $m$ vertices in $V^w$, with a slight abuse
of notations we denote them as $W_k$ for $0\le k\le m-1$. There are
$|C_1|\le 3n$ vertices in $V^s$, and we denote them as $s_j$ for
$1\le j\le |C_1|$. There is an edge between $W_k$ and $s_j$ if
$(w_i,s_j)\in T_2$ or $(\bar{w}_i,s_j)\in T_2$ where
$w_i,\bar{w}_i\in W_k$.

Consider the following problem: we want to draw each vertex of $V^s$ with a
color so that for any $0\le k\le m-1$, all the $s_j$ connected to
$W_k$ are drawn with different colors. There exists a greedy algorithm
for this problem which uses $\tau=O(n/m\log m)$ different colors.
Let $\chi(j)\le \tau$ be the color of $s_j$. We define the function
$f$ in the following way.
\vspace{-2.5mm}
\begin{itemize}
\setlength\itemsep{-1.5mm}
\item $f(b_{3kq+i})=i$, $f(a_{3kq+i})=3q+i$. Here $0\le k\le m-1$, $1\le i\le 3q=3n/m$
\item $f(w_{3kq+i})=6q+i$, $f(\bar{w}_{3kq+i})=9q+i$.
\item $f(s_j)=12q+\chi(j)$. Here $1\le \chi(j)\le \tau=O(n/m\log m)$.
\end{itemize}
\vspace{-2.5mm}
We show that $f$ satisfies the
'universal'\, property. For any proper matching, consider its matches containing elements of $W_k=\{w_i,\bar{w}_i|3kq+1\le i\le 3kq+3q\}$.
Let $Q$ be the set of elements in these matches. Then obviously for any $a_i\in Q$ or $b_i\in Q$, $3kq+1\le i\le 3kq+3q$,
implying that they are allocated with distinct bits. For any $s_j,s_{j'}\in Q$, our coloring implies that $\chi(j)\neq \chi(j')$ since
they both connected to $W_k$.
Thus elements of $Q$ are all allocated with distinct bits.
\vspace{-5mm}
\paragraph{Construction of the Scheduling Instance}
To further identify elements that are allocated with the same bit, we define function $g$ as:
\vspace{-2.5mm}
\begin{itemize}
\setlength\itemsep{-1.5mm}
\item $g(w_{3kq+i})=g(\bar{w}_{3kq+i})=g(a_{3kq+i})=g(b_{3kq+i})=m+k$.
\item Sort vertices with the same color in an arbitrary way. Suppose $s_j$ is colored with color $t$ and is the $l$-th vertex in the
sequence, then $g(s_j)=m+l-1$.
\end{itemize}
\vspace{-2.5mm}
We construct four kinds of jobs: a match job for every match, a
cover job for every element, dummy jobs and one huge job.

For every match $(w,x,y)$, we construct a job with processing time
$g(w)\alpha^{f(w)}+g(x)\alpha^{f(x)}+g(y)\alpha^{f(y)}$ where
$(w,x,y)$ may represent $(w_i)$ or $(\bar{w}_i)$ (in this case we
take $g(x)=g(y)=0$), or represent $(w_i,s_j)$ or $(\bar{w}_i,s_j)$
(in this case we take $g(y)=0$), or represent $(w_i,a_i,b_i)$ or
$(\bar{w}_i,a_i,b_{\zeta(i)})$.

For every element $\eta$, we construct a job with processing time
$(6m^3-g(\eta))\alpha^{f(\eta)}$ where $\eta$ may represent $w_i$,
$\bar{w}_i$, $s_j$, $a_i$ or $b_i$.

We construct dummy jobs.
Using the pigeonhole principle we can conclude that there are
$l_t\le m$ vertices colored with color $t$. If $l_t< m$, we then
construct $m-l_t$ dummy jobs, each of which has a processing time of
$6m^3\alpha^{12q+t}$.

Recall that there are $m+1$ machines, we construct a huge dummy job
whose processing time equals to
$6m^3(m+1)\sum_{i=1}^{9q+\tau}\alpha^{i}$ minus the total processing
time of all the jobs we construct before. It follows directly that
if there exists a feasible solution for $I_{sche}$ whose makespan is
no more than $6m^3\sum_{i=1}^{9q+\tau}\alpha^{i}$, the load of every
machine is $6m^3\sum_{i=1}^{9q+\tau}\alpha^{i}$.

\vspace{-5mm}
\paragraph{From 3DM' to Scheduling} Given that $f$ satisfies the 'universal'\, property,
it is not difficult to construct a schedule of makespan $6m^3\sum_{i=1}^{9q+\tau}\alpha^{i}$
based on a proper matching.
\vspace{-5mm}
\paragraph{From Scheduling to 3DM'} Suppose the huge job is on machine $m+1$, we focus on machine $1$ to $m$. It is easy to check that, for every $i$ there are only a constant number of jobs (except the huge job) with nonnegative
$\alpha^i$ term, and the coefficients are bounded by $6m^3$. Thus, by taking $\alpha=m^{O(1)}$ to be large enough, there
is no carry over when we add up all the jobs (except the huge job), implying that on machine $1$ to $m$, the coefficients of $\alpha^i$ terms from the jobs on each machine add up to $6m^3$. Recall the
definition of $g$. The (nonnegative) coefficient of $\alpha^i$ term from a match job is in $[m,2m]$, from a
cover job is in $[6m^3-m,6m^3-2m]$, and from a dummy job is $6m^3$, thus the $6m^3\alpha^i$ term in the load of each machine
is either contributed by a dummy job, or a cover job and a match job. Furthermore, if it is contributed by a cover job
and a match job, then the element corresponding to the cover job is contained in the match corresponding to the match job. We
can prove that all the cover jobs are on machine $1$ to $m$, given that there is one cover job for each element, the match jobs
on machine $1$ to $m$ correspond to a proper matching.

\clearpage

\appendix

\section{Scheduling on Arbitrary Number of Machines}
\label{appendix:scheduling}
\subsection{From SAT to SAT'}
\label{appendix:subsec 3SAT'} Given any instance $I_{sat}$ (with $m$
clauses) of the 3SAT problem, we may transform $I_{sat}$ into a
3SAT' instance $I'_{sat}$ in which every variable appears at most
three times. Such a transformation is due to Tovey and we describe
it as follows for the completeness.

Let $z$ be any variable in $I_{sat}$ and suppose it appears $d$
times in clauses. If $d=1$ then we add a dummy clause $(z\vee \neg
z)$. Otherwise $d\ge 2$ and we introduce $d$ new variables $z_1$,
$z_2$, $\cdots$, $z_d$ and $d$ new clauses $(z_1\vee\neg z_2)$,
$(z_2\vee\neg z_3)$, $\cdots$, $(z_{d}\vee \neg z_{1})$. Meanwhile
we replace the $d$ occurrences of $z$ by $z_1$, $z_2$, $\cdots$,
$z_d$ in turn and remove $z$. By doing so we transform $I_{sat}$
into $I_{sat}'$ by introducing at most $3m$ new variables and $3m$
new clauses.

Notice that each new clause we add in $I_{sat}'$ is of the form
$(z_i\vee\neg z_k)$. We let $C_2$ be the set of them and let $C_1$
be the set of other clauses. It is easy to verify that $I_{sat}'$ is
an instance of 3SAT' problem.

From now on we use $I_{sat}$ to denote a 3SAT' instance. Given any
positive $\delta>0$ (we assume $1/\delta\ge 2$ is an integer) and
$I_{sat}$ with $n$ variables, we construct an instance of the
scheduling problem such that it admits a feasible solution of
makespan $10^5r$ where $r=2^{3/\delta+9}n^{1+\delta}$ if and only if
$I_{sat}$ is satisfiable.

By adding dummy jobs we may assume that $n$ is sufficient large
(e.g., $n\ge 2^{3/\delta^2+7/\delta}$) and $n^{\delta}$ is an
integer. Recall that clauses could be divided into $C_1$ and $C_2$
where $|C_1|=n/3$ and $|C_2|=n$. All the variables and clauses in
$C_1$ could be re-indexed so that clause $c_i\in C_1$ contains
variables $z_i$, $z_{i+1}$ and $z_{i+2}$ for $i\in R=\{1,4,7,\cdots,
n-2\}$.

Next we define a set of functions $f_j$ and $g_j$ (corresponding to
the functions $f$ and $g$ when $\delta=1/2$) through recursively
partitioning the clauses in $C_2$.

\subsection{Partition Clauses}
We first partition all the clauses (of $C_2$) equally into
$n^{\delta}$ groups. Let these groups be $S_{1,k_1}$ for $1\le
k_1\le n^{\delta}$. We call them as layer-$1$ groups. It can be
easily seen that each layer-1 group contains exactly
$n_1=n^{1-\delta}$ clauses, as a consequence, clauses of $S_{1,k_1}$
contain $n_1$ positive literals and $n_1$ negative literals.

For simplicity, let $S_{1,k_1}^+$ be the indices of all the positive
literals of $S_{1,k_1}$ and $S_{1,k_1}^-$ be the indices of all the
negative literals of $S_{1,k_1}$.

Suppose $i^{(1,k_1)}_1<i^{(1,k_1)}_2<\cdots<i^{(1,k_1)}_{n_1}$ are
all the indices in $S_{1,k_1}^+$, we then define
$$f_{1/\delta}(i^{(1,k_1)}_l)=k_1,\quad g_{1/\delta-1}(i^{(1,k_1)}_{l})=l.$$

Similarly let
$\bar{i}^{(1,k_1)}_1<\bar{i}^{(1,k_1)}_2<\cdots<\bar{i}^{(1,k_1)}_{n_1}$
be all the indices in $S_{1,k_1}^-$, we then define
$$\bar{f}_{1/\delta}(\bar{i}^{(1,k_1)}_l)=k_1, \quad\bar{g}_{1/\delta-1}(\bar{i}^{(1,k_1)}_{l})=l.$$

Each group $S_{1,k_1}$ is then further partitioned equally into
$n^{\delta}$ subgroups and let these groups be $S_{2,k_1,k_2}$ for
$1\le k_2\le n^{1/\delta}$. In general, suppose we have already
derived $n^{(j-1)\delta}$ layer-$(j-1)$ groups for $2\le j\le
1/\delta-1$. Each layer-$(j-1)$ group, say,
$S_{j-1,k_1,k_2,\cdots,k_{j-1}}$ is then further partitioned equally
into $n^{\delta}$ subgroups. Let them be $S_{j,k_1,k_2,\cdots,k_j}$
for $1\le k_j\le n^{\delta}$. It can be easily seen that each
layer-$j$ group contains $n_j=n^{1-j\delta}$ clauses. Again let
$S_{j,k_1,k_2,\cdots,k_j}^+$ and $S_{j,k_1,k_2,\cdots,k_j}^-$ be the
sets of indices of all the positive literals and negative literals
in $S_{j,k_1,k_2,\cdots,k_j}$ respectively. Let
$i^{(j,k_1,k_2,\cdots,k_j)}_1<i^{(j,k_1,k_2,\cdots,k_j)}_2<\cdots<i^{(j,k_1,k_2,\cdots,k_j)}_{n_j}$
be the indices in $S_{j,k_1,k_2,\cdots,k_j}^+$, we then define
$$f_{1/\delta-j+1}(i^{(j,k_1,k_2,\cdots,k_j)}_l)=k_j,\quad g_{1/\delta-j}(i^{(j,k_1,k_2,\cdots,k_j)}_{l})=l.$$

Similarly let
$\bar{i}^{(j,k_1,k_2,\cdots,k_j)}_1<\bar{i}^{(j,k_1,k_2,\cdots,k_j)}_2<\cdots<\bar{i}^{(j,k_1,k_2,\cdots,k_j)}_{n_j}$
be all the indices in $\bar{S}_{j,k_1,k_2,\cdots,k_j}^-$, we then
define
$$\bar{f}_{1/\delta-j+1}(\bar{i}^{(j,k_1,k_2,\cdots,k_j)}_l)=k_j,
\quad\bar{g}_{1/\delta-j}(\bar{i}^{(j,k_1,k_2,\cdots,k_j)}_{l})=l.$$

The above procedure stops when we derive layer-$(1/\delta-1)$ groups
with each of them containing $n^{\delta}$ clauses. We have the
following simple observations.

\textbf{Observation}
\begin{itemize}
\item[1.]{For any $1\le i\le n$, $1\le f_k(i)\le
n^{\delta}$ for $2\le k\le 1/\delta$, and $1\le
g_k(i),\bar{g}_k(i)\le n^{k\delta}$ for $1\le k\le 1/\delta-1$.}
\item[2.]{If $(z_i\vee\neg z_h)\in C_2$, then $f_k(i)=\bar{f}_k(h)$ for $2\le k\le 1/\delta$.}
\item[3.]{For any $0\le k\le 1/\delta-2$ and $i< i'$}
\begin{itemize}
\item{If $f_{1/\delta}(i)=f_{1/\delta}(i')$,
$f_{1/\delta-1}(i)=f_{1/\delta-2}(i')$, $\cdots$,
$f_{1/\delta-k}(i)=f_{1/\delta-k}(i')$, \\then
$g_{1/\delta-k-1}(i)<g_{1/\delta-k-1}(i')$.}
\item{If $\bar{f}_{1/\delta}(i)=\bar{f}_{1/\delta}(i')$,
$\bar{f}_{1/\delta-1}(i)=\bar{f}_{1/\delta-2}(i')$, $\cdots$,
$\bar{f}_{1/\delta-k}(i)=\bar{f}_{1/\delta-k}(i')$, \\then
$\bar{g}_{1/\delta-k-1}(i)<\bar{g}_{1/\delta-k-1}(i')$.}
\end{itemize}
\item[4.]{For any $1\le\tau\le n^{\delta}$ and $2\le k\le 1/\delta$,
$|\{i|f_k(i)=\tau\}|=|\{i|\bar{f}_k(i)=\tau\}|=n^{1-\delta}$.}
\end{itemize}

\subsection{Construction of the Scheduling Instance}
We construct the scheduling instance based on $I_{sat}$. Throughout
this section we set $x=4n^{\delta}$, $r=2^{3/\delta+9}n^{1+\delta}$
and use $s(j)$ to denote the processing time of job $j$. We will
show that the constructed scheduling instance admits a feasible
schedule of makespan $K=10^5r$ if and only if the given 3SAT'
instance is satisfiable. Similar to the special case when
$\delta=1/2$, we construct $8n$ variable jobs, $8n$ truth assignment
jobs, $n$ clause jobs, $2n$ dummy jobs. The only difference is that
we construct more agent jobs, indeed, we will construct
$4(1/\delta-1)n$ agent jobs, divided from layer-1 agent jobs to
layer-$(1/\delta-1)$ agent jobs.

\noindent Variable jobs: $v_{i,1}^{\gamma}$ and $v_{i,2}^{\gamma}$
are constructed for $z_i$, $v_{i,3}^{\gamma}$ and $v_{i,4}^{\gamma}$
are for $\neg z_i$.
$$s(v_{i,k}^T)=r+2^{1/\delta+7}[f_{1/\delta}(i)x^{1/\delta}+i]+2^{1/\delta+6}+k,\quad
k=1,2$$
$$s(v_{i,k}^T)=r+2^{1/\delta+7}[\bar{f}_{1/\delta}(i)x^{1/\delta}+i]+2^{1/\delta+6}+k,\quad
k=3,4$$
$$s(v_{i,k}^F)=s(v_{i,k}^T)+2r,\quad k=1,2,3,4$$

\noindent Agent jobs: layer-$j$ agent jobs $\eta_{i,j,+}^{\gamma}$
and $\eta_{i,j,-}^{\gamma}$ are constructed for $1\le i\le n$ and
$1\le j\le 1/\delta-1$.
$$s(\eta_{i,j,+}^T)=r+2^{1/\delta+7}[\sum_{k=j+1}^{1/\delta}f_{k}(i)x^{k}+g_{j}(i)]+2^{j+6}+8,\quad j=1,2,\cdots, 1/\delta-1$$
$$s(\eta_{i,j,-}^T)=r+2^{1/\delta+7}[\sum_{k=j+1}^{1/\delta}\bar{f}_{k}(i)x^{k}+\bar{g}_{j}(i)]+2^{j+6}+16,\quad j=2,3,\cdots,1/\delta-1$$
Specifically,
$s(\eta_{i,1,-}^T)=r+2^{1/\delta+7}[\sum_{k=2}^{1/\delta}\bar{f}_{k}(i)x^{k}+\bar{g}_{1}(i)x]+2^{7}+16$,
$$s(\eta_{i,j,\sigma}^F)=s(\eta_{i,\sigma}^T)+2r, \quad\sigma=+,-$$

\noindent Truth assignment jobs: $a_i^{\gamma}$, $b_i^{\gamma}$,
$c_i^{\gamma}$ and $d_i^{\gamma}$.
$$s(a_i^F)=11r+(2^7i+8), s(b_i^F)=11r+(2^7i+32),$$
$$s(c_i^F)=101r+(2^7i+16), s(d_i^F)=101r+(2^7i+64),$$
$$s(k_i^{T})=s(k_i^F)+r,\quad k=a,b,c,d.$$

\noindent Clause jobs: 3 clause jobs are constructed for every
$c_{j}\in C_1$ where $j\in R$, with one $u_{j}^T$ and two copies of
$u_{j}^F$:
$$s(u_{j}^T)=10004r+2^{1/\delta+9}j, s(u_{j}^F)=10002r+2^{1/\delta+9}j.$$

\noindent Dummy jobs: $n+n/3$ jobs with processing time $1000r$, and
$n-n/3$ jobs with processing time $1002r$.

Let $V$ and $V_a$ be the set of variable jobs and agent jobs. Let
$A$, $B$, $C$, $D$ be the set of $a_i^{\gamma}$, $b_{i}^{\gamma}$,
$c_i^{\gamma}$ and $d_i^{\gamma}$ respectively. Let $G_0=V\cup V_a$,
$G_1=A\cup B$, $G_2=C\cup D$, $G_3$ be the set of dummy jobs and
$G_4=U$ be the set of clause jobs. Again we may drop the superscript
for simplicity. We construct huge jobs. They create gaps on
machines. According to which jobs are needed to fill up the gap,
they are divided into five groups.

Two huge jobs (variable-agent jobs) $\theta_{\eta,i,+}$ and
$\theta_{\eta,i,-}$ are constructed for each variable $z_i$:
\begin{eqnarray*}
s(\theta_{\eta,i,+})&=&10^5r-[4r+2^{1/\delta+7}(2f_{1/\delta}(i)x^{1/\delta}
+i+g_{1/\delta-1}(i))+2^{1/\delta+6}+2^{1/\delta+5}+10]\\
s(\theta_{\eta,i,-})&=&10^5r-[4r+2^{1/\delta+7}(2\bar{f}_{1/\delta}(i)x^{1/\delta}
+i+\bar{g}_{1/\delta-1}(i))+2^{1/\delta+6}+2^{1/\delta+5}+20]
\end{eqnarray*}

$2/\delta-4$ huge jobs (layer-decreasing jobs) $\theta_{i,j,+}$ and
$\theta_{i,j,-}$ are constructed for $j=1,\cdots,1/\delta-2$. For
$j=2,3,\cdots,1/\delta-2$, their processing times are
\begin{eqnarray*}
s(\theta_{i,j,+})&=&10^5r-[4r+2^{1/\delta+7}(2\sum_{k=j+2}^{1/\delta}f_{k}(i)x^{k}+f_{j+1}(i)x^{j+1}
+g_{j+1}(i)+g_{j}(i))+2^{j+7}+2^{j+6}+16]\\
s(\theta_{i,j,-})&=&10^5r-[4r+2^{1/\delta+7}(2\sum_{k=j+2}^{1/\delta}\bar{f}_{k}(i)x^{k}+\bar{f}_{j-1}(i)x^{j+1}
+\bar{g}_{j+1}(i)+\bar{g}_{j}(i))+2^{j+7}+2^{j+6}+32].
\end{eqnarray*}

For $j=1$, their processing times are
\begin{eqnarray*}
s(\theta_{i,1,+})&=&10^5r-[4r+2^{1/\delta+7}(2\sum_{l=3}^{1/\delta}f_{l}(i)x^{l}+f_{2}(i)x^{2}
+g_{2}(i)+g_{1}(i))+2^{8}+2^{7}+16]\\
s(\theta_{i,1,-})&=&10^5r-[4r+2^{1/\delta+7}(2\sum_{l=3}^{1/\delta}\bar{f}_{l}(i)x^{l}+\bar{f}_{2}(i)x^{2}
+\bar{g}_{2}(i)+\bar{g}_{1}(i)x)+2^{8}+2^{7}+32].
\end{eqnarray*}

One huge job (agent-agent job) $\theta_{i,k,C_2}$ is constructed for
$(z_i\vee \neg z_k)\in C_2$:
$$s(\theta_{i,k,C_2})=10^5r-[4r+2^{1/\delta+7}(2\sum_{l=2}^{1/\delta}f_{l}(i)x^{l}
+\bar{g}_{1}(k)x+g_{1}(i))+2^{8}+24].$$

Three huge jobs (variable-clause jobs) are constructed for each
$c_{j}\in C_1$ ($j\in R$), one for each literal: for $i=j,j+1,j+2$,
if $z_i\in c_j$, we construct $\theta_{j,i,+,C_1}$, otherwise $\neg
z_i\in c_j$, and we construct $\theta_{j,i,-,C_1}$.
$$s(\theta_{j,i,+,C_1})=10^5r-11005r-(2^{1/\delta+7}f_{1/\delta}(k)x^{1/\delta}+2^{1/\delta+9}i+2^{1/\delta+7}k+2^{1/\delta+6}+1),$$
$$s(\theta_{j,i,-,C_1})=10^5r-11005r-(2^{1/\delta+7}\bar{f}_{1/\delta}(k)x^{1/\delta}+2^{1/\delta+9}i+2^{1/\delta+7}k+2^{1/\delta+6}+3).$$

One huge job (variable-dummy job) is constructed for each variable.
Notice that each variable appears exactly three times in clauses, if
$z_i$ appears twice while $\neg z_i$ appears once, we construct
$\theta_{i,-}$. Otherwise, we construct $\theta_{i,+}$ instead.
$$s(\theta_{i,+})=10^5r-1003r-(2^{1/\delta+7}f_{1/\delta}(i)x^{1/\delta}+2^{1/\delta+7}i+2^{1/\delta+6}+1),$$
$$s(\theta_{i,-})=10^5r-1003r-(2^{1/\delta+7}\bar{f}_{1/\delta}(i)x^{1/\delta}+2^{1/\delta+7}i+2^{1/\delta+6}+3).$$

Thus, for each clause $c_i$ ($i\in R$) and $k=i,i+1,i+2$, either
$\theta_{i,+}$ and $\theta_{j,i,-,C_1}\in \Theta_1$ exist, or
$\theta_{i,-}$ and $\theta_{j,i,+,C_1}$ exist.

Four huge jobs (variable-assignment jobs) are constructed for each
variable $z_i$, namely $\theta_{i,a,c}$, $\theta_{i,b,d}$,
$\theta_{i,a,d}$ and $\theta_{i,b,c}$:
$$s(\theta_{i,a,c})=10^5r-115r-2^{1/\delta+7}(f_{1/\delta}(i)x^{1/\delta}+i)-(2^{1/\delta+6}+2^8i+25),$$
$$s(\theta_{i,b,d})=10^5r-115r-2^{1/\delta+7}(f_{1/\delta}(i)x^{1/\delta}+i)-(2^{1/\delta+6}+2^8i+98),$$
$$s(\theta_{i,a,d})=10^5r-115r-2^{1/\delta+7}(\bar{f}_{1/\delta}(i)x^{1/\delta}+i)-(2^{1/\delta+6}+2^8i+75),$$
$$s(\theta_{i,b,c})=10^5r-115r-2^{1/\delta+7}(\bar{f}_{1/\delta}(i)x^{1/\delta}+i)-(2^{1/\delta+6}+2^8i+52).$$

The jobs we construct now are similar to that we construct in the
special case, except that we construct a set of agent jobs from
layer-1 to layer-$(1/\delta-1)$ instead of only two agent jobs, and
a set of layer-decreasing jobs so as to leave gaps for these agent
jobs. It is easy to verify that we construct $2/\delta n+5n$ huge
jobs, and thus there are $2/\delta n+5n$ identical machines in the
scheduling instance.

The processing time of each job is a polynomial on $x$. The reader
may refer to the following tables for an overview of the
coefficients

\begin{table}[!hbp]
\begin{tabular}{|c|c|c|c|c|c|c|c|c|}
\hline Jobs/coefficients &$2^{1/\delta+7}x^{1/\delta}$ &
$2^{1/\delta+7}x^{1/\delta-1}$ & $\cdots$ & $2^{1/\delta+7}x^{j+1}$
&$2^{1/\delta+7}x^{j}$ &$2^{1/\delta+7}x^{j-1}$&$\cdots$& $2^{1/\delta+7}x^{2}$ \\
\hline
$\eta_{i,j,+}$ & $f_{1/\delta}(i)$ & $f_{1/\delta-1}(i)$ & $\cdots$ & $f_{j+1}(i)$& $0$&0&$\cdots$&0 \\
\hline
$\eta_{i,j-1,+}$ & $f_{1/\delta}(i)$ & $f_{1/\delta-1}(i)$ & $\cdots$ & $f_{j+1}(i)$& $f_{j}(i)$&0&$\cdots$&0 \\
\hline
\end{tabular}
\caption{coefficients-of-agent-jobs}\label{fig:coe-of-link}
\end{table}

\begin{table}[!hbp]
\begin{tabular}{|c|c|c|c|c|c|c|c|c|}
\hline Jobs/coefficients &$2^{1/\delta+7}x^{1/\delta}$ &
$2^{1/\delta+7}x^{1/\delta-1}$ & $\cdots$ & $2^{1/\delta+7}x^{j+1}$
&$2^{1/\delta+7}x^{j}$ &$2^{1/\delta+7}x^{j-1}$&$\cdots$& $2^{1/\delta+7}x^{2}$ \\
\hline
$\theta_{i,j,+}$ & $2f_{1/\delta}(i)$ & $2f_{1/\delta-1}(i)$ & $\cdots$ & $f_{j+1}(i)$& $0$&0&$\cdots$&0 \\
\hline
$\theta_{i,j-1,+}$ & $2f_{1/\delta}(i)$ & $2f_{1/\delta-1}(i)$ & $\cdots$ & $2f_{j+1}(i)$& $f_{j}(i)$&0&$\cdots$&0 \\
\hline $\theta_{i,k,C_2}$ & $2f_{1/\delta}(i)$ &
$2f_{1/\delta-1}(i)$ & $\cdots$ & $2f_{j+1}(i)$&
$2f_{j}(i)$&$2f_{j-1}(i)$
&$\cdots$&$2f_{2}(i)$ \\
\hline
\end{tabular}
\caption{coefficients-of-huge-jobs}\label{fig:coe-of-relation}
\end{table}

It is not difficult to verify that the total processing time of all
the jobs is $(2/\delta n+5n)\cdot 10^5r$.

\subsection{3SAT to Scheduling}
\label{appendix-subsection: sat to schedule} We show that, if
$I_{sat}$ is satisfiable, then the makespan of the optimal solution
for the constructed scheduling instance is $10^5r$. Notice that the
number of huge jobs equals the number of machines. We put one huge
job on each machine. For simplicity, we may use the symbol of a huge
job to denote a machine, e.g., we call a machine as machine
$\theta_{i,a,c}$ if the job $\theta_{i,a,c}$ is on it.

We first schedule jobs in the following way. Recall that the
superscript ($T$ or $F$) of a job only influences its $r$-term. It
is not difficult to verify that by scheduling jobs in the following
way, except for the $r$-terms, the coefficients of other terms of
jobs on the same machine add up to $0$.

\begin{figure}[!h]
\begin{center}
      \includegraphics[scale=0.25]{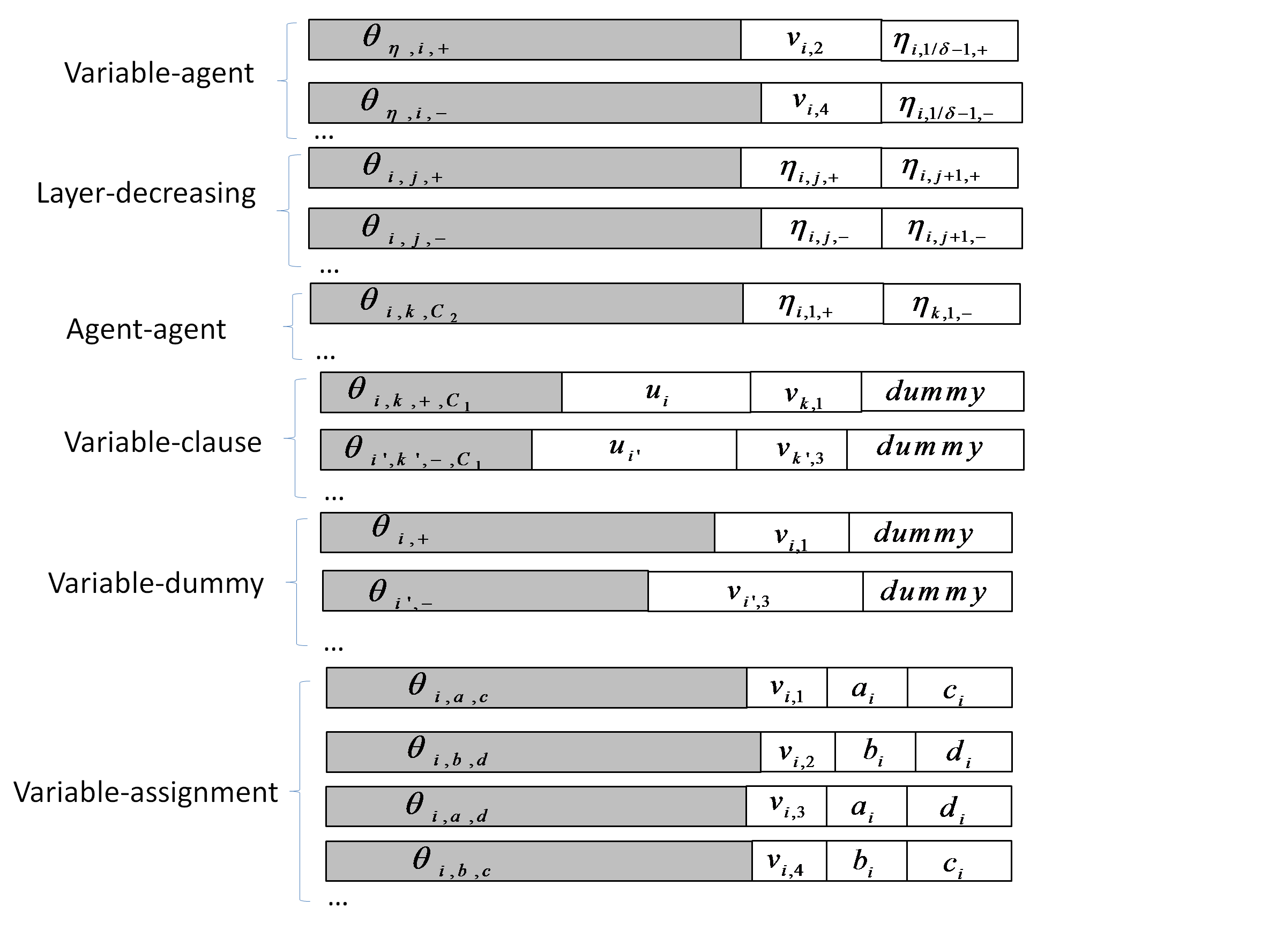}
  \end{center}
\label{fig:index-scheduling}\caption{index-scheduling}
\end{figure}

We determine the superscripts of each job so that their $r$-terms
add up to $10^5r$. Suppose according to the truth assignment of
$I_{sat}$ variable $z_i$ is true, we determine the superscripts in
the following way.

On machine $\theta_{\eta,i,+}$, the two jobs are $v_{i,2}^T$ and
$\eta_{i,1/\delta-1,+}^F$. On machine $\theta_{i,j,+}$ where
$j=1,2,\cdots,1/\delta-2$, the two agent jobs are $\eta_{i,j,+}^F$
and $\eta_{i,j+1,+}^T$. Thus, $\eta_{i,j,+}^F$ is on machine
$\theta_{i,j,+}$, $\eta_{i,j,+}^T$ is on machine $\theta_{i,j-1,+}$,
which means both the true job and false job of $\eta_{i,j,+}$ are
scheduled. While for $\eta_{i,1,+}$, only $\eta_{i,1,+}^F$ is
scheduled (on machine $\theta_{i,1,+}$).

Similarly on machine $\theta_{\eta,i,-}$, the two jobs are
$v_{i,4}^F$ and $\eta_{i,1/\delta-1,-}^T$. On machine
$\theta_{i,j,-}$ where $j=1,2,\cdots,1/\delta-2$, the two variable
jobs are $\eta_{i,j,-}^T$ and $\eta_{i,j+1,-}^F$. Thus,
$\eta_{i,j,-}^T$ is on machine $\theta_{i,j,-}$, $\eta_{i,j,-}^F$ is
on machine $\theta_{i,j-1,-}$, which means both the true job and
false job is scheduled. While for $\eta_{i,1,-}$, only
$\eta_{i,1,-}^T$ is scheduled (on machine $\theta_{i,1,-}$).

Consider agent-agent machines. For the variable $z_i$, there is a
clause $(z_i\vee\neg z_k)\in C_2$ for some $k$, and we know
$\eta_{i,1,+}$ and $\eta_{k,1,-}$ are on machine $\theta_{i,k,C_2}$.
Since $I_{sat}$ is satisfiable, variables $z_i$ and $z_k$ should be
both true or both false. Thus, given that $z_i$ is true, $z_k$ is
also true. This implies that $\eta_{i,1,+}^T$ and $\eta_{k,1,-}^F$
are not scheduled before, and we let the two jobs on machine
$\theta_{i,k,C_2}$ be them.

Meanwhile in $C_2$ there is also a clause $(z_{k'}\vee\neg z_i)$ for
some $k'$, and we know $\eta_{i,1,-}$ and $\eta_{k',1,+}$ are on
machine $\theta_{k',i,C_2}$. Since $I_{sat}$ is satisfiable,
variables $z_{k'}$ and $z_i$ should be both true or both false.
Thus, given that $z_i$ is true, $z_k$ is also true. This implies
that $\eta_{k',1,+}^T$ and $\eta_{i,1,-}^F$ are not unscheduled
before, and we let the two jobs on machine $\theta_{k',i,C_2}$ be
them.

Consider variable-assignment jobs. We put $v_{i,1}^F$, $a_i^F$,
$c_i^F$ on machine $\theta_{i,a,c}$, put ${v}_{i,2}^F,b_i^F,d_i^F$
on machine $\theta_{i,b,d}$, put $v_{i,3}^T,a_i^T,d_i^T$ on machine
$\theta_{i,a,d}$, and put $v_{i,4}^T,b_i^T,c_i^T$ on machine
$\theta_{i,b,c}$. Thus, both the true copy and false copy of $a_i$,
$b_i$, $c_i$ and $d_i$ are scheduled. It can be easily seen that the
$r$-terms of three true jobs or three false jobs both add up to
$115r$. Otherwise, $z_i$ is false, and we schedule jobs just in the
opposite way, i.e., we replace each true job with its corresponding
false job, and each false job with its corresponding true job in the
previous scheduling.

We consider the remaining jobs. If $z_i$ is true, then $v_{i,1}^T$
and $v_{i,3}^F$ are left. If $z_i$ is false, then $v_{i,1}^F$ and
$v_{i,3}^T$ are left. These jobs should be scheduled with clause
jobs and dummy jobs on variable-clause machines or variable-dummy
machines. Notice that for any $i\in R$ and $k\in \{i,i+1,i+2\}$,
either $\theta_{i,k,+,C_1}$ and $\theta_{k,-}$ exist, or
$\theta_{i,k,-,C_1}$ and $\theta_{k,+}$ exist.

Suppose the variable $z_k$ is true. If $\theta_{i,k,+,C_1}$ and
$\theta_{k,-}$ exist, we put $v_{k,1}^T$ on machine
$\theta_{i,k,+,C_1}$, and $v_{k,3}^F$ on machine $\theta_{k,-}$.
Otherwise $\theta_{i,k,-,C_1}$ and $\theta_{k,+}$ exist, and we put
$v_{k,3}^F$ on machine $\theta_{i,k,-,C_1}$, and $v_{k,1}^T$ on
machine $\theta_{k,+}$.

In both cases, the remaining jobs $v_{k,1}^T$ and $v_{k,3}^F$ are
scheduled. Otherwise $z_{k}$ is false. If $\theta_{i,k,+,C_1}$ and
$\theta_{k,-}$ exist, put $v_{k,1}^F$ on machine
$\theta_{i,k,+,C_1}$, and $v_{k,3}^T$ on machine $\theta_{k,-}$.
Otherwise $\theta_{i,k,-,C_1}$ and $\theta_{k,+}$ exist. We put
$v_{k,3}^T$ on machine $\theta_{i,k,-,C_1}$, and $v_{k,1}^F$ on
machine $\theta_{k,+}$.

Again in both cases, the remaining jobs $v_{k,1}^F$ and $v_{k,3}^T$
are scheduled. From now on we drop the symbol $+$ or $-$ and just
use $\theta_{i,k,C_1}$ to denote either $\theta_{i,k,+,C_1}$ or
$\theta_{i,k,-,C_1}$, and use $\theta_{k}$ to denote either
$\theta_{k,+}$ or $\theta_{k,-}$. It is easy to verify that the
above scheduling has the following property.

\vskip 2mm\noindent\textbf{Property} If $c_i$ is satisfied by
variable $z_k$ (i.e., $z_k\in c_i$ and $z_k$ is true or $\neg z_k\in
c_i$ and $z_k$ is false), then a true variable job is on machine
$\theta_{i,k,C_1}$; if $c_i$ is not satisfied by $z_k$, then a false
variable job is on machine $\theta_{i,k,C_1}$.

Consider variable-dummy machines. For each $k=1,2,\cdots,n$, there
is one machine $\theta_{k}$. If a true variable job is on it, we
then put additionally a dummy job of size $1002r$. Otherwise a false
variable job is on it, and we put additionally a dummy job of size
$1000r$ on it. Thus, in both cases the $r$-terms of variable job and
dummy job add up to $1003r$.

Consider variable-clause machines. For each clause $c_i\in C_1$
(i.e., $i\in R$), there are three copies of $u_i$, one true and two
false. There are three machines, $\theta_{i,i,C_1}$,
$\theta_{i,i+1,C_1}$ and $\theta_{i,i+2,C_1}$.

Notice that according to the truth assignment, $c_i$ is satisfied by
at least one variable. Suppose $c_i$ is satisfied by $z_{k_1}$, and
let $z_{k_2}$ and $z_{k_3}$ be the remaining two variables in this
clause, i.e., $k_1,k_2,k_3$ is some permutation of the three indices
$i,i+1,i+2$. We put $u_i^T$ on machine $\theta_{i,k_1,C_1}$.
Additionally, we put a dummy job of size $1000r$ on this machine.
According to the property we have mentioned above, since $c_i$ is
satisfied by $z_{k_1}$, the variable job on machine
$\theta_{i,k_1,C_1}$ is a true job. Thus, the $r$-terms of the true
clause job, true variable job and a dummy job on
$\theta_{i,k_1,C_1}$ add up to $11005r$.

Consider machine $\theta_{i,k_2,C_1}$ and $\theta_{i,k_3,C_1}$. We
put one of the remaining two false jobs $u_i^F$ on them
respectively. We add dummy jobs according to the following criteria.
If the variable job is true, we add a dummy job of size $1002r$. If
the variable job is false, we add a dummy job of size $1000r$.

Thus in both cases, the $r$-terms of the variable job and dummy job
add up to $1003r$. And if we further add the $r$-terms of the false
clause job and the relation job, the sum is $10^5r$. Finally we
check the number of dummy jobs that are used.

For simplicity we use $(T/F,T/F,1000r/1002r)$ to denote the
truth-type of a variable-clause machine, i.e., the first coordinate
is $T$ is the variable job is true, and $F$ if it is false,
similarly the second coordinate is $T$ (or $F$) if the clause job is
$T$ (or $F$), the third coordinate is $1000r$ (or $1002r$) if the
dummy job is of size $1000r$ (or $1002r$). We also denote the
truth-type of a variable-dummy machine in the form of
$(T/F,1000r/1002r)$.

A dummy job of size $1000r$ is always scheduled on a machine of
truth-type $(T,T,1000r)$, $(F,F,1000r)$ and $(F,1000r)$, while a
dummy job of $1002r$ is scheduled on a machine of truth-type
$(T,F,1002r)$ and $(T,1002r)$. Notice that on these machines, there
are $n$ true variable jobs and $n$ false variable jobs, and there
are $|C_1|=n/3$ true clause jobs, thus simple calculations show that
$n+n/3$ dummy jobs of $1000r$ and $n-n/3$ dummy jobs of $1002r$ are
scheduled, which coincides with the dummy jobs we construct.

\subsection{Scheduling to 3SAT}
We show that, if the constructed scheduling instance admits a
feasible schedule with makespan $10^5r$, then $I_{sat}$ is
satisfiable. Notice that in a scheduling problem, jobs are
represented by their processing times rather than symbols, we first
show that we can the processing time of each job we construct is
distinct (except that two copies of $u_i^F$ are constructed for
every clause in $C_1$), this would be enough to determine the symbol
of a job from its processing time.

\subsubsection{Distinguishing Jobs from Their Processing Times}
Recall that we define the processing time of a job in the form of a
polynomial, we use the notion $x^j$-term or $r$-term in their direct
meaning. Meanwhile, we call the sum of all except the $r$-term of a
job as the small-$r$-term. For any $2\le j\le 1/\delta$, we delete
the $r$-term and $x^k$-term with $k\ge j$ from the processing time
of a job, and call the sum of all the remaining terms as the
small-$x^j$-term.

For example, the relation job $\theta_{i,3,+}$ is of processing time
$10^5r-[4r+2^{1/\delta+7}(2\sum_{k=5}^{1/\delta}f_{k}(i)x^{k}+f_{4}(i)x^{4}
+g_{4}(i)+g_{3}(i))+2^{10}+2^{9}+16]$, and thus for $5\le k\le
1/\delta$, its $x^k$-term is $2^{1/\delta+7}\cdot 2f(i)x^k$. Its
$x^4$-term is $f(4)(i)x^4$. Its $x^3$-term and $x^2$-term are 0. Its
small-$x^5$-term is $2^{1/\delta+7}(f_{4}(i)x^{4}
+g_{4}(i)+g_{3}(i))+2^{10}+2^{9}+16$. Meanwhile, for a clause job,
say, $u_i$, its $x^j$-term is 0 for $1\le j\le 1/\delta$, and its
small-$x^j$-term for any $j$ is $2^{1/\delta+9}i$.

Consider the small-$r$-term of any job. If it is a huge job, this
value is negative and its absolute value is bounded by
$2^{1/\delta+7}(2\sum_{k=2}^{1/\delta}n^{1/\delta}x^k+2n)+2^{1/\delta+7}+32<1/2r$
(notice that $n^{\delta}x^k= 1/4 x^{k+1}$). Otherwise it is a
variable, or agent, or clause, or truth assignment, or dummy job,
and the sum is positive with its absolute value also bounded by
$2^{1/\delta+7}(\sum_{k=2}^{1/\delta}n^{1/\delta}x^k+n)+2^{1/\delta+6}+64<1/4r$.

For the small-$x^j$-terms of jobs, we have the following lemma.

\begin{lemma}
\label{le:small x^j} For a huge job, its small-$x^j$-term ($2\le
j\le 1/\delta$) is negative, and the absolute value is bounded by
$2^{1/\delta+7}\cdot3/4x^j$. For a variable or agent job, its
small-$x^j$-term is positive and bounded by
$2^{1/\delta+7}\cdot3/8x^j$.
\end{lemma}
\begin{proof}
Notice that $g_j(i)\le n^{j\delta}$, while $f_j(i)x^j\ge
2^{2j}n^{j\delta}>g_j(i)$ for any $2\le j\le 1/\delta-1$. Thus for a
huge job, its small-$x^j$-term is at most
\begin{eqnarray*}
&&2^{1/\delta+7}[2\sum_{l=2}^{j-1}f_{l}(i)x^l+\bar{g}_1(k)x+g_1(i)]+2^{1/\delta+6}+2^{1/\delta+5}+32\\
&\le &
2^{1/\delta+7}[2\sum_{l=2}^{j-1}n^{\delta}x^l+n^{\delta}x+n^{\delta}+1]\\
&\le & 2^{1/\delta+7}[2\sum_{l=2}^{j-1}n^{\delta}x^l+2n^{\delta}x]\\
&\le & 2^{1/\delta+7}[2\sum_{l=3}^{j-1}n^{\delta}x^l+3n^{\delta}x^2]\\
&\le & 2^{1/\delta+7}[2\sum_{l=4}^{j-1}n^{\delta}x^l+3n^{\delta}x^3]\\
&\cdots&\\
&\le & 2^{1/\delta+7}\cdot 3n^{\delta}x^{j-1}\\
&\le & 2^{1/\delta+7}\cdot 3/4x^{j}
\end{eqnarray*}
The inequalities make use of the simple observation that
$n^{\delta}x^k=1/4x^{k+1}$ for $1\le k\le 1/\delta$. The proof for
variable or agent jobs is similar.
\end{proof}

Given the processing time of a job, we can easily determine whether
it is a huge, variable, agent, clause, or dummy job by considering
its quotient of divided by $r$, and the residual of divided by
$2^7$, and if it is a huge job, we may further determine if it is a
variable-agent, layer-decreasing, agent-agent, variable-clause,
variable-dummy, variable-assignment job. Using the above lemma, if
it is a variable, or agent, or clause, or dummy job, we can easily
expand it into the summation form and determine its symbol according
to Observation 3.

Suppose we are given the processing time of a huge jobs. Again it is
easy to determine its symbol if it is a variable-assignment,
variable-dummy or variable-clause job. If it is an agent-agent job,
then according to the fact that $g_1(i)\le n^{\delta}\le 1/4x$, we
can also expand the processing time into the summation form and
determine its symbol. If it is a variable-agent or layer-decreasing
job, we show that the processing time of such a job is unique.

Suppose $s(\theta_{i_1,j_1,+})=s(\theta_{i_2,j_2,+})$, then
according to Lemma \ref{le:small x^j} we have $j_1=j_2=j$ and
$f_{k}(i_1)=f_k(i_2)$ for $j+1\le k\le 1/\delta$ and
$g_{j+1}(i_1)+g_{j}(i_1)=g_{j+1}(i_2)+g_{j}(i_2)$. Now according to
Observation 3, we have $i_1=i_2$. Similarly if
$s(\theta_{i_1,j_1,-})=s(\theta_{i_2,j_2,-})$, we can also prove
that $i_1=i_2$, $j_1=j_2$. Obviously it is impossible that
$s(\theta_{i_1,j_1,+})=s(\theta_{i_2,j_2,-})$. The proof for
variable-agent jobs is similar.

\subsubsection{Scheduling to 3SAT}
We prove the following lemma.
\begin{lemma}
\label{le:sche->3sat'}If there is a solution for the constructed
scheduling instance in which the load of each machine is $10^5r$,
then $I_{sat}$ is satisfiable.
\end{lemma}

Let $Sol^*$ be an optimal solution, it can be easily seen that there
is a huge job on each machine, leaving a gap if the load of each
machine is $10^5r$. We may use the symbol of a huge job to denote
the corresponding gap and the machine it is scheduled on.

We divide jobs into groups based on their processing times.
According to the previous subsection, we know the processing time of
a variable or agent job is either in $[r,5/4r]$ or in $[3r,13/4r]$.
Let $G_0$ be the set of them. The processing time of $a_i$ or $b_i$
belongs to $[11r,12.5r]$, of $c_i$ or $d_i$ belongs to
$[101r,102.5r]$. Let $G_1=A\cup B$, $G_2=C\cup D$.

\begin{lemma}
\label{le:sche-3sat'-step 1} In $Sol^*$, besides the huge job, the
other jobs on a machine are:
\begin{itemize}
\item{The variable-agent, or layer-decreasing, or agent-agent gap is filled up by two jobs of $G_0$.}
\item{The variable-clause gap is filled up by one clause job, one dummy job and one job of $G_0$.}
\item{The variable-dummy gap is filled up by one dummy job and one job of $G_0$.}
\item{The variable-assignment gap is filled up by one job of $G_1=A\cup B$,
one job of $G_2=C\cup D$, and one job of $G_0$.}
\end{itemize}
\end{lemma}
\begin{proof}
See the following table (Table~\ref{fig:gaps size}) as an overview
of gaps on machines (here $\Theta_0$ denotes the set of
variable-agent, layer-decreasing and agent-agent gaps).

\begin{table}[!hbp]
\begin{tabular}{|c|c|c|c|c|}
\hline Machines(Gaps) &$\Theta_0$ &
Variable-clause & Variable-dummy & Variable-assignment\\
\hline
Size of Gaps & $(4r,5r)$ & $(11005r,11006r)$ & $(1003r,1004r)$ & $(115r,116r)$\\
\hline
\end{tabular}
\label{fig:gaps size}
\end{table}

Consider clause jobs. According to the table they can only be used
to fill variable-clause gaps. Meanwhile each variable-clause machine
(gap) could accept at most one clause job. Notice that there are $n$
clause jobs and $n$ variable-clause machines, thus there is one
clause job on every variable-clause machine. By further subtracting
the processing time of the clause job from the gap, the remaining
gap of a variable-clause machine belongs to $[1000r,1004r]$.

Consider dummy jobs. According to the current gaps, they can only be
scheduled on variable-clause or variable-dummy machines, and each of
these machines could accept at most one dummy job. Again notice that
there are $2n$ such machines and $2n$ dummy jobs, there is one dummy
job on every variable-clause and variable-dummy machine. The current
gap of a variable-clause machine is in $[0,4r]$, of a variable-dummy
machine is in $[r,4r]$. Using the same argument we can show that
there is one job of $C\cup D$ and one job of $A\cup B$ on each
variable-assignment machine.

Consider variable and agent jobs. Each machine of $\Theta_0$ has a
gap in $(4r,5r)$, implying that there are at least two variable or
agent jobs on it. The current gap of a variable-assignment machine
is at least $115r-(102r+2^7n+12r+2^7n+64+64)\ge r-2^{9}n> 1/2r$,
thus there is at least one variable or agent job on it. Similarly
there is at least one variable or agent job on a variable-dummy
machine.

Consider each variable-clause machine. As we have determined, there
are a clause and a dummy job on it. We check their total processing
times more carefully. By subtracting the huge job in from $10^5r$,
the gap is in $[11005r,(11005+1/2)r]$. If the clause job on this
machine is a true job, with a processing time over $10004r$, then
the dummy job on it can only be of $1000r$, otherwise the total
processing time of the two jobs is over $11006r$, which is a
contradiction. Thus, the total processing time of the two jobs is at
most $11004r+2^{1/\delta+9}n+1000r\le (11004+1/2)r$, which means
there is at least one variable or agent job on this machine.
Otherwise, the clause job on this machine is a false job with a
processing time at most $10002r+2^{1/\delta+9}n\le (10002+1/2)r$.
Adding a dummy job, their total processing time is at most
$(11004+1/2)r$, and again we can see that there is at least one
variable or agent job on this machine.

The above analysis shows that there is at least one job of $G_0$ on
a variable-clause, variable-dummy and variable-assignment machine,
and at least two jobs of $G_0$ on each machine of $\Theta_0$,
requiring $4n+4/\delta n$ jobs, which equals to $|G_0|$. Thus the
lemma follows directly.
\end{proof}

Given the above lemma, we consider the residuals of each job divided
by $2^{1/\delta+7}$. The fact that the three or four residuals on
each machine should add up to $0$ implies the following table (Table
5).

\begin{table}[hbp]
\begin{center}
\begin{tabular}{|c|c|c|c|}
\hline $\theta_{\eta,i,+}$ & $v_{i',2}$ & $\eta_{i'',1/\delta-1,+}$ & $\backslash$\\
\hline
$\theta_{\eta,i,-}$ & $v_{i',4}$ & $\eta_{i'',1/\delta-1,-}$ & $\backslash$\\
\hline
$\theta_{i,j,+}$ & $\eta_{i',j,+}$ & $\eta_{i'',j+1,+}$ & $\backslash$\\
\hline
$\theta_{i,j,-}$ & $\eta_{i',j,-}$ & $\eta_{i'',j+1,-}$ & $\backslash$\\
\hline
$\theta_{i,k,C_2}$ & $\eta_{i',1,+}$ & $\eta_{i'',1,-}$ & $\backslash$\\
\hline
$\theta_{i,k,+,C_1}$ & $u_{i'}$ & $v_{i'',1}$ & dummy\\
\hline
$\theta_{i,k,-,C_1}$ & $u_{i'}$ & $v_{i'',3}$ & dummy\\
\hline
$\theta_{i,+}$ & $v_{i',1}$ & dummy & $\backslash$\\
\hline
$\theta_{i,-}$ & $v_{i',3}$ & dummy & $\backslash$\\
\hline
$\theta_{i,a,c}$ & $v_{i_1,1}$ & $a_{i_2}$ & $c_{i_3}$\\
\hline
$\theta_{i,b,d}$ & $v_{i_1,2}$ & $b_{i_2}$ & $d_{i_3}$\\
\hline
$\theta_{i,a,d}$ & $v_{i_1,3}$ & $a_{i_2}$ & $d_{i_3}$\\
\hline
$\theta_{i,b,c}$ & $v_{i_1,4}$ & $a_{i_2}$ & $d_{i_3}$\\
\hline
\end{tabular}
\label{fig:structure}\caption{Structure}
\end{center}
\end{table}

The next step is to characterize the indices, i.e., we need to prove
that for each row, $i=i'=i''$ (or $i=i_1=i_2=i_3$). If the indices
equal for jobs on a machine, this machine (gap) is called satisfied.
The above table, combined with Lemma~\ref{le:small x^j}, implies the
following lemma.

\begin{lemma}
\label{le:r and delta equal} For jobs on each machine, their
$r$-terms add up to $10^5r$, $x^k$-terms ($2\le k\le 1/\delta-1$)
add up to $0$.
\end{lemma}

The $x^k$-term of each huge job is negative and should be canceled
by the corresponding terms from other jobs. Similar as the proof for
the special case when $\delta=1/2$, we would divide the $x^k$-terms
($2\le k\le 1/\delta$) of each huge job (gap) into singular terms
and regular terms. Notice that here we use the notion of singular
(regular) terms instead of singular (regular) gaps because when
$1/\delta>2$ we need to consider multiple terms of a gap.

We define singular (regular) terms in the following way. The
$x^{1/\delta}$-terms of variable-clause, variable-dummy and
variable-assignment gaps are singular terms.

For other gaps, see Table~\ref{fig:regular}. The terms marked with
$*$ are singular term (e.g., the $x^j$-term of
$\theta_{i,j-1,\sigma}$), all the other terms are regular terms.

\begin{table}[hbp]
\begin{tabular}{|c|c|c|c|c|c|c|c|c|}
\hline Gaps/Coefficients &$2^{1/\delta+7}x^{1/\delta}$ &
$2^{1/\delta+7}x^{1/\delta-1}$ & $\cdots$ & $2^{1/\delta+7}x^{j+1}$
&$2^{1/\delta+7}x^{j}$ &$2^{1/\delta+7}x^{j-1}$&$\cdots$& $2^{1/\delta+7}x^{2}$ \\
\hline
$\theta_{\eta,i,+}$ & $2f_{1/\delta}(i)$ & 0 & $\cdots$ & 0& 0&0&$\cdots$&0 \\
\hline
$\theta_{\eta,i,-}$ & $2\bar{f}_{1/\delta}(i)$ & 0 & $\cdots$ & 0& 0&0&$\cdots$&0 \\
\hline
$\theta_{i,j-1,+}$ & $2f_{1/\delta}(i)$ & $2f_{1/\delta-1}(i)$ & $\cdots$ & $2f_{j+1}(i)$& $f_j(i)^*$&0&$\cdots$&0 \\
\hline
$\theta_{i,j-1,-}$ & $2\bar{f}_{1/\delta}(i)$ & $2\bar{f}_{1/\delta-1}(i)$ & $\cdots$ & $2\bar{f}_{j+1}(i)$& $\bar{f}_{j}(i)^*$&0&$\cdots$&0 \\
\hline $\theta_{i,k,C_2}$ & $2f_{1/\delta}(i)$ &
$2f_{1/\delta-1}(i)$ & $\cdots$ & $2f_{j+1}(i)$&
$2f_{j}(i)$&$2f_{j-1}(i)$
&$\cdots$&$2f_{2}(i)$ \\
\hline
\end{tabular}
\caption{Singular and regular terms}\label{fig:regular}
\end{table}

A singular term of a gap, say, $2^{1/\delta+7}\tau x^j$ for $1\le
\tau\le n^{\delta}$, is called well-canceled, if it is filled up by
one job with the $x^j$-term of $2^{1/\delta+7}\tau x^j$ and other
jobs with the $x^j$-terms of $0$. A regular term, say,
$2^{1/\delta+7}\cdot 2\tau x^j$ for $1\le \tau\le n^{\delta}$, is
called well-canceled, if it is filled up by two jobs whose
$x^j$-terms are $2^{1/\delta+7}\tau x^j$.

\begin{lemma}
\label{le:singular well-cancel}Every singular term is well-canceled.
\end{lemma}
The proof is straightforward.

\begin{lemma}
\label{le:regular well-cancel}Every regular term is well-canceled.
\end{lemma}
Before we prove this lemma, we first count the number of variable
and agent jobs whose $x^k$-term is $2^{1/\delta+7}\cdot \tau_kx^k$
where $2\le k\le 1/\delta$ and $1\le \tau_k\le n^{\delta}$. For
simplicity we call them as $\tau_k$-jobs. According to Observation
4,
$|\{i|f_k(i)=\tau_k\}|=|\{i|\bar{f}_k(i)=\tau_k\}|=n^{1-\delta}=n_1$,
thus we have Table \ref{fig:number variable}.

\begin{table}[!hbp]
\begin{tabular}{|c|c|c|c|c|c|c|}
\hline Jobs/Coefficients &$2^{1/\delta+7}x^{1/\delta}$ &
$2^{1/\delta+7}x^{1/\delta-1}$ & $\cdots$
&$2^{1/\delta+7}x^{j}$ &$\cdots$& $2^{1/\delta+7}x^{2}$ \\
\hline
$v_{i,\iota}$($\iota=1,2,3,4$) & $f_{1/\delta}(i),\bar{f}_{1/\delta}(i)$ & 0 & $\cdots$ & $0$ &$\cdots$&0 \\
\hline
$\eta_{i,1/\delta-1,+},\eta_{i,1/\delta-1,-}$ & $f_{1/\delta}(i),\bar{f}_{1/\delta}(i)$ & 0 & $\cdots$ & $0$ &$\cdots$&0 \\
\hline
$\cdots$&$\cdots$&$\cdots$&$\cdots$&$\cdots$&$\cdots$&$\cdots$\\
\hline $\eta_{i,j-1,+},\eta_{i,j-1,-}$ &
$f_{1/\delta}(i),\bar{f}_{1/\delta}(i)$ &
$f_{1/\delta-1}(i),\bar{f}_{1/\delta-1}(i)$
 & $\cdots$ & $f_{j}(i),\bar{f}_{j}(i)$ &$\cdots$&0 \\
 \hline
$\cdots$&$\cdots$&$\cdots$&$\cdots$&$\cdots$&$\cdots$&$\cdots$\\
\hline $\eta_{i,1,+},\eta_{i,1,-}$ &
$f_{1/\delta}(i),\bar{f}_{1/\delta}(i)$ &
$f_{1/\delta-1}(i),\bar{f}_{1/\delta-1}(i)$
 & $\cdots$ & $f_{j}(i),\bar{f}_{j}(i)$ &$\cdots$&$f_{2}(i),\bar{f}_{2}(i)$ \\
\hline  $\sharp$ $\tau_k$-jobs &
$2(2n_1/\delta+2n_1)$&$2\times2(1/\delta-2)n_1$
 & $\cdots$ & $2\times 2(j-1)n_1$ &$\cdots$&$2\times 2n_1$ \\
\hline
\end{tabular}
\caption{Counting numbers of variable and agent jobs}
\label{fig:number variable}
\end{table}

The factor $2$ in the last row comes from the fact that for each
symbol there are actually a true job and a false job, and thus the
numbers should double. We call the gap whose $x^k$-term is a regular
term and equals to $2^{1/\delta+7}\cdot 2\tau_kx^k$ as a regular
$\tau_k$-gaps, and call the gap whose $x^k$-term is a singular term
and equals to $2^{1/\delta+7}\cdot \tau_kx^k$ as a singular
$\tau_k$-gap. We count their numbers. See Table~\ref{fig:number
huge} as an overview.

\begin{table}[!hbp]
\begin{tabular}{|c|c|c|c|c|c|c|}
\hline Gaps/Coefficients &$2^{1/\delta+7}x^{1/\delta}$ &
$2^{1/\delta+7}x^{1/\delta-1}$ & $\cdots$
&$2^{1/\delta+7}x^{j}$ &$\cdots$& $2^{1/\delta+7}x^{2}$ \\
\hline
$\theta_{i,1/\delta-1,+},\theta_{i,1/\delta-1,-}$ & $2f_{1/\delta}(i),2\bar{f}_{1/\delta}(i)$ & 0 & $\cdots$ & $0$ &$\cdots$&0 \\
\hline
$\cdots$&$\cdots$&$\cdots$&$\cdots$&$\cdots$&$\cdots$&$\cdots$\\
\hline $\theta_{i,j-1,+},\theta_{i,j-1,-}$ &
$2f_{1/\delta}(i),2\bar{f}_{1/\delta}(i)$ &
$2f_{1/\delta-1}(i),2\bar{f}_{1/\delta-1}(i)$
 & $\cdots$ & $f_{j}(i),\bar{f}_{j}(i)$ &$\cdots$&0 \\
 \hline
$\cdots$&$\cdots$&$\cdots$&$\cdots$&$\cdots$&$\cdots$&$\cdots$\\
\hline $\theta_{i,1,+},\theta_{i,1,-}$ &
$2f_{1/\delta}(i),2\bar{f}_{1/\delta}(i)$ &
$2f_{1/\delta-1}(i),2\bar{f}_{1/\delta-1}(i)$
 & $\cdots$ & $2f_{j}(i),2\bar{f}_{j}(i)$ &$\cdots$&$f_{2}(i),\bar{f}_{2}(i)$ \\
\hline  $\sharp$ singular $\tau_k$-gaps & $6n_1$&$2n_1$
 & $\cdots$ & $2n_1$ &$\cdots$&$2n_1$ \\
\hline
 $\sharp$ regular $\tau_k$-gaps & $2n_1/\delta-n_1$&$2n_1(1/\delta-1)-3n_1$
 & $\cdots$ & $2jn_1-3n_1$ &$\cdots$&$n_1$ \\
 \hline
\end{tabular}
\caption{Count the number of gaps} \label{fig:number huge}
\end{table}

Notice that in Table \ref{fig:number huge} we do not list
variable-clause, variable-dummy and variable-assignment gaps,
however, they contribute to the number of singular
$2^{1/\delta+7}\tau_{1/\delta}x^{1/\delta}$ terms by $6n_1$ for any
$1\le \tau_{1/\delta}\le n^{\delta}$. Now we come to the proof of
Lemma \ref{le:regular well-cancel}.

\begin{proof}
We prove the lemma through induction. We first consider
$x^{1/\delta}$-terms. A regular $x^{1/\delta}$-term of a gap could
always be expressed as $2^{1/\delta+7}\cdot
2\tau_{1/\delta}x^{1/\delta}$ for $1\le \tau_{1/\delta}\le
n^{\delta}$.

We start with $\tau_{1/\delta}=1$. Notice that a regular
$x^{1/\delta}$-term comes from a variable-agent, layer-decreasing or
agent-agent gap. According to Table 5, the $x^{1/\delta}$-term of
the other two jobs (variable or agent jobs) used to fill up such a
gap are nonzero and at least $2^{1/\delta+7}x^{1/\delta}$, thus the
regular term $2^{1/\delta+7}\cdot 2x^{1/\delta}$ is well canceled.
Suppose for any $\tau_{1/\delta}< h_0\le n^{\delta}$, each regular
term $2^{1/\delta+7}\tau_{1/\delta}x^{1/\delta}$ is well-canceled.

We consider the case that $\tau_{1/\delta}= h_0$. For any
$\tau_{1/\delta}$ such that $1\le \tau_{1/\delta} <h_0$, there are
in all $4n_1(1/\delta+1)$ variable or agent jobs whose
$x^{1/\delta}$-term is $2^{1/\delta+7}\tau_{1/\delta}x^{1/\delta}$
(see Table~\ref{fig:number variable}). We determine the scheduling
of these jobs.

Among them $6n_1$ jobs are on used to cancel singular terms
according to Lemma \ref{le:singular well-cancel}. Meanwhile since
there are $2n_1/\delta-n_1$ gaps with regular terms
$2^{1/\delta+7}\cdot 2\tau_{1/\delta}x^{1/\delta}$ (see
Table~\ref{fig:number huge}), the induction hypothesis implies that
$4n_1/\delta-2n_1$ of these variable and agent jobs are used to
cancel these regular terms.

Thus, we can conclude that for a regular $x^{1/\delta}$-term being
$2^{1/\delta+7}\cdot 2h_0x^{1/\delta}$, both of the $x^{1/\delta}$
term of the two jobs (variable or agent jobs) used to cancel it are
at least $2^{1/\delta+7}\cdot h_0x^{1/\delta}$. This implies, again,
that the regular term $2^{1/\delta+7}\cdot 2h_0x^{1/\delta}$ is
well-canceled. The proof for regular $x^k$-terms are the same.
\end{proof}

Next we prove that in $Sol^*$, every machine is satisfied. See
Figure \ref{fig:index-scheduling} as an illustration of such a
solution. Obviously a variable-dummy machine (gap) is satisfied.

\begin{lemma}
\label{le:Theta_0^3} Agent-agent machines (gaps) are satisfied.
\end{lemma}
\begin{proof}
Consider each agent-agent machine, say, $\theta_{i_0,k_0,C_2}$. We
can assume that the other two jobs on it are $\eta_{i,1,+}$ and
$\eta_{k,1,-}$. Then according to Lemma \ref{le:regular
well-cancel}, we have
\begin{eqnarray*}
&&f_l(i)=\bar{f}_l(k)=f_l(i_0)=\bar{f}_l(k_0),\quad
l=2,3,\cdots,1/\delta\\
&&g_1(i)+\bar{g}_1(k)x=g_1(i_0)+\bar{g}_1(k_0)x.
\end{eqnarray*}

Since $x=4n^{\delta}$, while
$g_1(i),g_1(i_0),\bar{g}_1(k),\bar{g}_1(k_0)\le n^{\delta}$, thus
$g_1(i)=g_1(i_0)$, $\bar{g}_1(k)=\bar{g}_1(k_0)$. According to the
construction of functions $f$ and $g$ (see Observation 3), we know
that $i=i_0$ and $k=k_0$.
\end{proof}

We consider variable-clause machines. Notice that for each $i_0$ and
$k_0\in\{i_0,i_0+1,i_0+2\}$, either $\theta_{i_0,k_0,+,C_1}$ or
$\theta_{i_0,k_0,-,C_1}$ exists.

\begin{lemma}
\label{le:clause job with variable} Machine $\theta_{1,k,+,C_1}$ or
$\theta_{1,k,-,C_1}$ ($k=1,2,3$) is satisfied. The machine
$\theta_{i_0,k_0,+,C_1}$ or $\theta_{i_0,k_0,-,C_1}$ for $i_0\ge 2$
and $k_0\in\{i_0,i_0+1,i_0+2\}$ is satisfied if:
\begin{itemize}
\item{For $i<i_0$, each machine $\theta_{i,k,+,C_1}$ or $\theta_{i,k,-,C_1}$ is satisfied.}
\item{All variable jobs $v_{k',\iota}$ with $k'< i_0$ and $\iota=1,2,3,4$ are not scheduled on this
machine.}
\end{itemize}
\end{lemma}
\begin{proof}
We consider clause $c_{1}\in C_1$. As $c_{1}$ contains three
variables $z_1$, $z_2$ and $z_3$, there are three huge jobs
$\theta_{1,1,\sigma_1,C_1}$, $\theta_{1,2,\sigma_2,C_1}$ and
$\theta_{1,3,\sigma_3,C_1}$ where
$\sigma_1,\sigma_2,\sigma_3\in\{+,-\}$. Meanwhile there are three
clause jobs of $u_1$.

For $i_0=1$ and any $k_0\in\{1,2,3\}$, suppose
$\theta_{1,k_0,+,C_1}$ exists, and the two jobs together with it are
a clause job $u_{i}$ and a variable job $v_{k,\iota}$ with
$\iota\in\{1,2,3,4\}$. Since
$s(\theta_{1,k_0,+,C_1})=10^5r-11005r-(2^{1/\delta+7}f_{1/\delta}(1)+2^{1/\delta+9}+2^{1/\delta+7}k_0+2^{1/\delta+6}+1)$,
according to Lemma \ref{le:r and delta equal}, we have
$2^{1/\delta+9}i+2^{1/\delta+7}k+1=2^{1/\delta+9}+2^{1/\delta+7}k_0+\iota$.
If $i\ge 2$, then the left side is at least $2^{1/\delta+10}$, while
the right side is at most $2^{1/\delta+9}+2^{1/\delta+7}\times
3+4<2^{1/\delta+10}$, which is a contradiction. Thus $i=1$ and it
follows directly that $k=k_0$, $\iota=1$. Otherwise
$\theta_{1,k_0,-,C_1}$ exists, and the proof is just similar. Thus,
machine $\theta_{1,k_0,+,C_1}$ or $\theta_{1,k_0,-,C_1}$
($k_0=1,2,3$) is satisfied.

When $i_0\ge 2$ and $k_0\in\{i_0+1,i_0+2,i_0+3\}$, again we suppose
that $\theta_{i_0,k_0,+,C_1}$ exists. Notice that for any $i\le
i_0-1$, $c_{i}$ contains three variables. According to the
hypothesis, the three clause jobs $u_{i}$ are scheduled on three
machines, they are $\theta_{i,i,+,C_1}$ or $\theta_{i,i,-,C_1}$,
$\theta_{i,i+1,+,C_1}$ or $\theta_{i,i+1,-,C_1}$ and
$\theta_{i,i+2,+,C_1}$ or $\theta_{i,i+2,-,C_1}$. Thus when we
consider machine $\theta_{i_0,k_0,+,C_1}$, all clause jobs $u_{i}$
with $i\le i_0-1$ could not be scheduled on this machine.

Again suppose that the two jobs scheduled together with
$\theta_{i_0,k_0,+,C_1}$ are $u_{i'}$ and $v_{k',\iota}$, then
$2^{1/\delta+9}i_0+2^{1/\delta+7}k_0+1=2^{1/\delta+9}i'+2^{1/\delta+7}k'+\iota$.
Since $i'\ge i_0-1$ and $k'\ge i'$, if $i'\ge i_0+1$, then we have
$2^{1/\delta+9}i'+2^{1/\delta+7}k'+\sigma>2^{1/\delta+9}(i_0+1)+2^{1/\delta+7}(i_0+1)\ge
2^{1/\delta+9}i_0+2^{1/\delta+7}(i_0+3)+1$, which is a
contradiction. Thus $i'=i_0$, $k'=k_0$ and $\iota=1$, which means
machine $\theta_{i_0,k_0,+,C_1}$ is satisfied.

Similarly if $\theta_{i_0,k_0,-,C_1}$ exists, this machine is also
satisfied.
\end{proof}

\begin{lemma}
\label{le:truth assignment with variable} Machines $\theta_{1,a,c}$,
$\theta_{1,b,d}$, $\theta_{1,a,d}$ and $\theta_{1,b,c}$ are
satisfied. Moreover, machines $\theta_{i_0,a,c}$,
$\theta_{i_0,b,d}$, $\theta_{i_0,a,d}$ and $\theta_{i_0,b,c}$ for
$i_0\ge 2$ are satisfied if:
\begin{itemize}
\item{Machines $\theta_{i,a,c}$,
$\theta_{i,b,d}$, $\theta_{i,a,d}$ and $\theta_{i,b,c}$ are
satisfied for $i\le i_0-1$.}
\item{All variable jobs $v_{i',\iota}$ with
$i'< i_0$ and $\iota\in\{1,2,3,4\}$ are not scheduled on these
machines.}
\end{itemize}
\end{lemma}
\begin{proof}
Consider machine $\theta_{1,a,c}$. Except the huge job, let the
other three jobs be $v_{i_1,\iota}$ ($\iota\in\{1,2,3,4\}$),
$a_{i_2}$ and $c_{i_3}$. Then we have
$$2^{1/\delta+7}i_1+2^{1/\delta+6}+\iota_1+2^7i_2+8+2^7i_3+16=2^{1/\delta+7}+2^{1/\delta+6}+2^8+25.$$
It can be easily seen that $i_1=i_2=i_3=1$ and $\iota=1$. Thus,
machine $\theta_{1,a,c}$ is satisfied. Using similar arguments we
can show that machines $\theta_{1,b,c}$, $\theta_{1,a,d}$ and
$\theta_{1,b,d}$ are satisfied.

The proof that machines $\theta_{i_0,a,c}$, $\theta_{i_0,b,d}$,
$\theta_{i_0,a,d}$ and $\theta_{i_0,b,c}$ are satisfied for $i_0\ge
2$ if two conditions of the lemma hold is the same.
\end{proof}

For simplicity, we call variable jobs $v_{i,\iota_1}$ with
$\iota_1\in\{1,2,3,4\}$ and agent jobs $\eta_{i,j,\iota_2}$ with
$\iota_2\in\{+,-\}$ as jobs of index-level $i$.

In contrast, let $\sigma\in\{+,-\}$, we call machine
$\theta_{\eta,i,\sigma}$, $\theta_{i,j,\sigma}$, machine
$\theta_{i',i,\sigma,C_1}$, machine $\theta_{i,\sigma}$, machines
$\theta_{i,a,c}$, $\theta_{i,a,d}$, $\theta_{i,b,c}$,
$\theta_{i,b,d}$ as machines of index-level $i$.

Specifically, machine $\theta_{i,k,C_2}$ is of index-level $i$ and
also of index-level $k$, i.e., this machine would appear in the set
of machines with index-level of $i$ as well as the set of machines
with index-level of $k$. Notice that according to Lemma
\ref{le:Theta_0^3} these machines are already satisfied.

\begin{lemma}
\label{le:each satisfy} In $Sol^*$, every machine (gap) is
satisfied.
\end{lemma}

\begin{proof}
We prove it through induction on the index-level of machines. We
start with $i=1$.

Consider machine $\theta_{\eta,1,+}$. We assume jobs $v_{i,2}$ and
$\eta_{i',1/\delta-1,+}$ are on it. Then simple calculations show
that
$$2f_{1/\delta}(1)x^{1/\delta}+1+g_{1/\delta-1}(1)=f_{1/\delta}(i)x^{1/\delta}+i+f_{1/\delta}(i')x^{1/\delta}+g_{1/\delta-1}(i').$$

According to Lemma \ref{le:regular well-cancel},
$f_{1/\delta}(1)=f_{1/\delta}(i)=f_{1/\delta}(i')$.

Since $i'\ge 1$, according to Observation 3 we have
$g_{1/\delta-1}(i')\ge g_{1/\delta-1}(1)$. Meanwhile $i\ge 1$, thus
$g_{1/\delta-1}(i')= g_{1/\delta-1}(1)$ and $i=1$. Again, due to
Observation 3 we have $i=i'=1$. Thus $v_{1,2}$ and
$\eta_{1,1/\delta-1,+}$ are on machine $\theta_{\eta,1,+}$, i.e.,
this machine is satisfied. Similarly we can prove that $v_{1,4}$ and
$\eta_{1,1/\delta-1,-}$ are on machine $\theta_{\eta,1,-}$.

Consider machine $\theta_{1,j,+}$ for $1\le j\le 1/\delta-2$. We
assume jobs $\eta_{i,j,+}$ and $\eta_{i',j+1,+}$ are on it. Then
simple calculations show that
$$2\sum_{l=j+2}^{1/\delta}f_{l}(1)x^l+f_{j+1}(1)x^{j+1}+g_{j+1}(1)+g_{j}(1)=
\sum_{l=j+1}^{1/\delta}f_{l}(i)x^l+\sum_{l=j+2}^{1/\delta}f_l(i')x^l+g_{j+1}(i')+g_{j}(i).$$

According to Lemma \ref{le:regular well-cancel}, we have
$$f_l(i)=f_l(1), \quad l=j+1,j+2,\cdots, 1/\delta,$$
$$f_l(i')=f_l(1),\quad l=j+2,j+3,\cdots, 1/\delta.$$
Thus $g_{j+1}(1)+g_{j}(1)=g_{j+1}(i')+g_{j}(i).$

According to Observation 3, we have $g_{j}(i)\ge g_{j}(1)$ and
$g_{j+1}(i')\ge g_{j+1}(1)$. Thus $g_j(i)=g_j(1)$,
$g_{j+1}(i')=g_{j+1}(1)$. Again due to Observation 3 we have
$i=i'=1$, i.e., machine $\theta_{1,j,+}$ is satisfied.

Similarly we can prove that machine $\theta_{1,j,-}$ for $2\le j\le
1/\delta-2$ is also satisfied. For $j=1$, recall that there is a
slight difference between $\theta_{1,1,-}$ and $\theta_{1,1,+}$, we
prove that machine $\theta_{1,1,-}$ is satisfied separately.

Consider $\theta_{1,1,-}$ and assume jobs $\eta_{i,1,-}$ and
$\eta_{i',2,-}$ are on it. Then
$$2\sum_{l=3}^{1/\delta}\bar{f}_{l}(1)x^l+\bar{f}_{2}(1)x^2+\bar{g}_{2}(1)+\bar{g}_{1}(1)x=
\sum_{l=2}^{1/\delta}\bar{f}_{l}(i)x^l+\sum_{l=3}^{1/\delta}\bar{f}_l(i')x^l+\bar{g}_{2}(i')+\bar{g}_{1}(i)x.$$

According to Lemma \ref{le:regular well-cancel}, we have
$$\bar{f}_l(i)=\bar{f}_l(1), \quad l=2,3,\cdots, 1/\delta,$$
$$\bar{f}_l(i')=f_l(1),\quad l=3,4,\cdots, 1/\delta.$$
Thus
$\bar{g}_{2}(1)+\bar{g}_{1}(1)x=\bar{g}_{2}(i')+\bar{g}_{1}(i)x$.
Similarly due to observation 3 we have $\bar{g}_{2}(i')\ge
\bar{g}_{2}(1)$, and $\bar{g}_{1}(i)\ge \bar{g}_{1}(1)$. Thus again
we can prove $i=i'=1$, which implies that machine $\theta_{1,1,-}$
is also satisfied.

Combining Lemma \ref{le:clause job with variable}, Lemma
\ref{le:truth assignment with variable} and Lemma
\ref{le:Theta_0^3}, we have proved so far that each machine of
index-level $1$ is satisfied. We further show that indeed, all the
variable and agent jobs of index-level $1$ are on machines of
index-level $1$. To see why, see Figure \ref{fig:sche-index1} for an
overview of the scheduling of jobs of index-level 1 (here Case 1
means $z_1\in C_1$, while Case 2 means $\neg z_1\in C_1$).

\begin{figure}[!h]
\begin{center}
\includegraphics[scale=0.25]{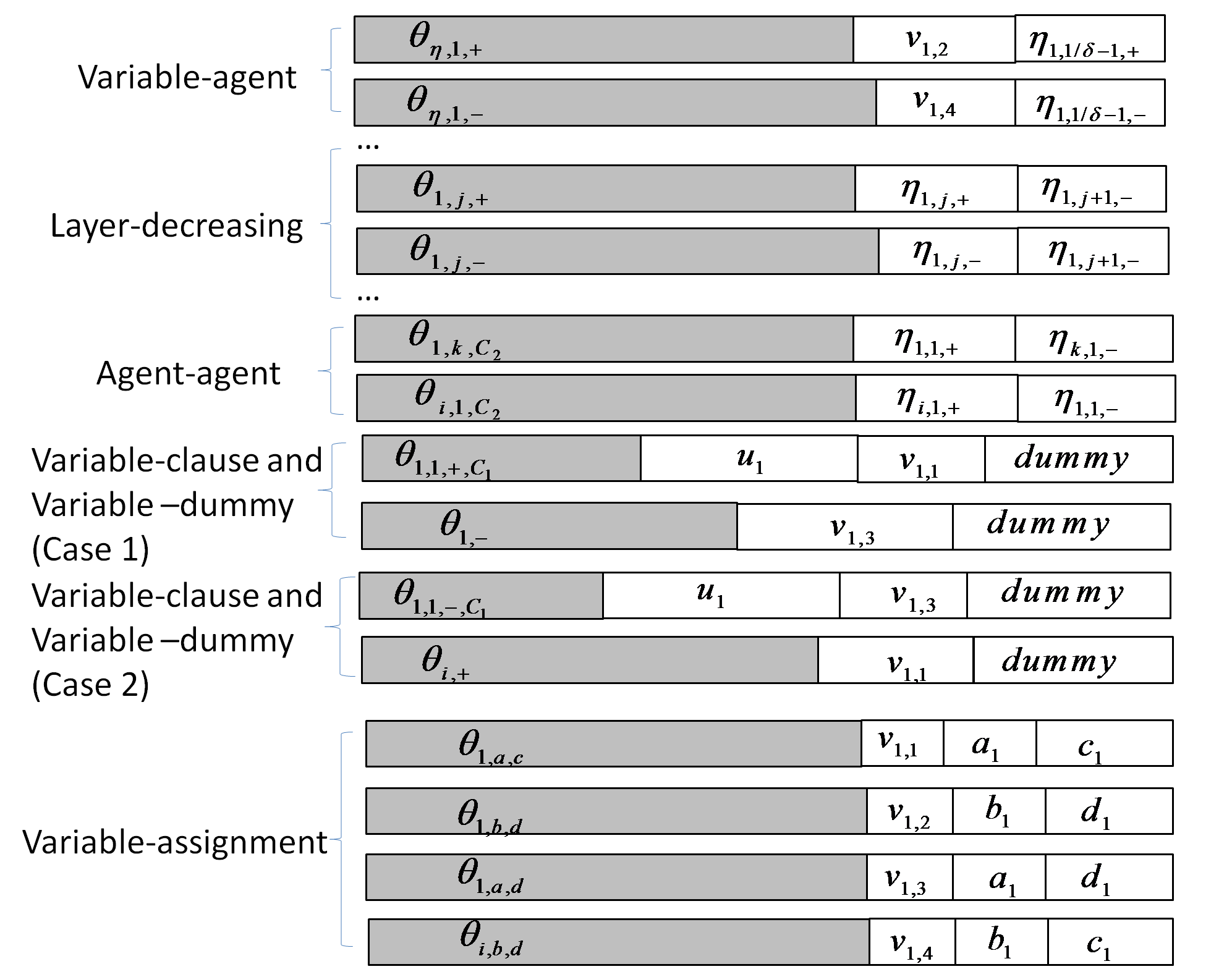}
\end{center}
\caption{scheduling-of-indexlevel-1}\label{fig:sche-index1}
\end{figure}

Suppose that for any $i< i_0\le n$, each machine of index-level $i$
is satisfied and all the variable or agent jobs of index-level $i$
are on machines of index-level $i$. We consider $i=i_0$.

According to Lemma \ref{le:clause job with variable} and Lemma
\ref{le:truth assignment with variable}, we know that machines
$\theta_{i_0,k,+,C_1}$ (or $\theta_{i_0,k,-,C_1}$) for $k\in
\{i_0,i_0+1,i_0+2\}$ and machines $\theta_{i_0,a,c}$,
$\theta_{i_0,b,d}$, $\theta_{i_0,a,d}$, $\theta_{i_0,b,c}$ are
satisfied.

Consider machine $\theta_{\eta,i_0,+}$ which is of index-level
$i_0$. Again we may assume jobs $v_{i,2}$ and
$\eta_{i',1/\delta-1,+}$ are on it, and the induction hypothesis
implies that $i\ge i_0$, $i'\ge i_0$. Simple calculations show that
$$2f_{1/\delta}(i_0)x^{1/\delta}+i_0+g_{1/\delta-1}(i_0)=f_{1/\delta}(i)x^{1/\delta}+i+f_{1/\delta}(i')x^{1/\delta}+g_{1/\delta-1}(i').$$

According to Lemma \ref{le:regular well-cancel},
$f_{1/\delta}(i_0)=f_{1/\delta}(i)=f_{1/\delta}(i')$. Since $i'\ge
i_0$, according to Observation 3 we have $g_{1/\delta-1}(i')\ge
g_{1/\delta-1}(i_0)$. Meanwhile $i\ge i_0$, thus
$g_{1/\delta-1}(i')= g_{1/\delta-1}(i_0)$ and $i=i_0$. We can
conclude that $i=i'=i_0$. So, $v_{i_0,2}$ and
$\eta_{i_0,1/\delta-1,+}$ are on machine $\theta_{\eta,i_0,+}$,
i.e., this machine is satisfied. Similarly we can prove that
$v_{i_0,4}$ and $\eta_{i_0,1/\delta-1,-}$ are on machine
$\theta_{\eta,i_0,-}$.

Consider machine $\theta_{i_0,j,+}$ for $1\le j\le 1/\delta-2$. We
assume jobs $\eta_{i,j,+}$ and $\eta_{i',j+1,+}$ are on it. Then
simple calculations show that
$$2\sum_{l=j+2}^{1/\delta}f_{l}(i_0)x^l+f_{j+1}(i_0)x^{j+1}+g_{j+1}(i_0)+g_{j}(i_0)=
\sum_{l=j+1}^{1/\delta}f_{l}(i)x^l+\sum_{l=j+2}^{1/\delta}f_l(i')x^l+g_{j+1}(i')+g_{j}(i).$$

According to Lemma \ref{le:regular well-cancel}, we have
$$f_l(i)=f_l(i_0), \quad l=j+1,j+2,\cdots, 1/\delta,$$
$$f_l(i')=f_l(i_0),\quad l=j+2,j+3,\cdots, 1/\delta.$$
Thus $g_{j+1}(i_0)+g_{j}(i_0)=g_{j+1}(i')+g_{j}(i).$

According to the hypothesis we know $i,i'\ge i_0$. Due to
Observation 3, we have $g_{j}(i)\ge g_{j}(i_0)$ and $g_{j+1}(i')\ge
g_{j+1}(i_0)$. Thus $g_j(i)=g_j(i_0)$, $g_{j+1}(i')=g_{j+1}(i_0)$,
which implies again that $i=i'=i_0$, i.e., machine
$\theta_{i_0,j,+}$ is satisfied.

Similarly we can prove that machine $\theta_{i_0,j,-}$ for $1\le
j\le 1/\delta-2$ is also satisfied (again we need to prove machine
$\theta_{i_0,1,-}$ is satisfied separately, and the proof is
actually the same as the case when $i_0=1$).

The above analysis shows that each machine of index-level $i_0$ is
satisfied. Similar to the case when $i_0=1$, we can further show
that all the variable and agent jobs of index-level $i_0$ are on
machines of index-level $i_0$.
\end{proof}

A machine is called truth benevolent if one of the following three
conditions holds.
\begin{itemize}
\item{For a variable-agent, layer-decreasing or agent-agent machine, the two jobs (variable or agent) on it are one true and one false.}
\item{For a variable-clause machine, The variable and clause job on
it are of the form $(T,T)$, $(F,F)$ or $(T,F)$.}
\item{For a variable-assignment machine, the variable and truth-assignment jobs on it are of the form
$(F,F,F)$ or $(T,T,T)$.}
\end{itemize}

We have the following lemma.

\begin{lemma}
\label{le:truth benevolent} In $Sol^*$, every machine of is truth
benevolent.
\end{lemma}
\begin{proof}
Consider a variable-agent, layer-decreasing or agent-agent machine.
On each of these machines, the $r$-terms of the two (variable or
agent) jobs should add up to $4r$ according to Lemma \ref{le:r and
delta equal}, thus the two jobs are one true and one false.

Consider a variable-clause machine. We check the $r$-terms of the
clause, variable and dummy job. According to Lemma \ref{le:r and
delta equal}, there are three possibilities that the three $r$-terms
add up to $11005r$, which are $r+10004r+1000r$, $3r+10002r+1000r$
and $r+10002r+1002r$, thus the variable and clause jobs are always
of the form $(T,T)$, $(F,F)$ or $(T,F)$.

Consider variable-assignment machines. We check the $r$-terms.
Except for the huge job, the $r$-terms of the variable job, $a_i$ or
$b_i$, $c_i$ or $d_i$ should add up to $115r$ and thus there are
only two possibilities, $r+12r+102r$ and $3r+11r+101r$, which
implies that they are of the form $(F,F,F)$ or $(T,T,T)$.
\end{proof}

Now we come to the proof of Lemma \ref{le:sche->3sat'}.
\begin{proof}
We assign values to variables according to the variable-assignment
machines. For each $i$, consider the four machines,
$\theta_{i,a,c}$, $\theta_{i,b,d}$, $\theta_{i,a,d}$ and
$\theta_{i,b,c}$. Since the three jobs are $(T,T,T)$ or $(F,F,F)$,
thus $a_i^T$ is on the same machine with either $c_i^T$ or $d_i^T$.

If $a_i^T$ is scheduled with $c_i^T$, then the jobs on the two
machines with $\theta_{i,a,c}$ and $\theta_{i,b,d}$ are
$(v_{i,1}^T,a_i^T,c_i^T)$, $(v_{i,2}^T,b_i^T,d_i^T)$. We let
variable $z_i$ be false. Otherwise $a_i^T$ is scheduled with
$d_i^T$, and the jobs on the two machines with $\theta_{i,a,d}$ and
$\theta_{i,b,c}$ are $(v_{i,3}^T,a_i^T,d_i^T)$ and
$(v_{i,4}^T,b_i^T,c_i^T)$. We let variable $z_i$ be true. We show
that every clause is satisfied.

For each $c_j\in C_1$, there is one job $u_j^T$, and it should be
scheduled with a true variable job. If it is $v_{i,1}^T$ where
$i=j,j+1$ or $j+2$, then it turns out that $z_i$ is true because
otherwise $v_{i,1}^T$ is already scheduled with $a_i^T$ and $c_i^T$.
Notice that either machine $\theta_{j,i,+,C_1}$ or machine
$\theta_{j,i,-,C_1}$ exists. Since $v_{i,1}$ is scheduled with
$u_j$, machine $\theta_{j,i,-,C_1}$ does not exist because otherwise
$v_{i,3}$, instead of $v_{i,1}$, is scheduled together with $u_j$ on
this machine. Thus the huge job $\theta_{j,i,+,C_1}$ exists, which
means the positive literal $z_i$ appears in $c_j$, thus $c_j$ is
satisfied. Otherwise it is $v_{i,3}^T$, then it turns out that $z_i$
is false. As $v_{i,3}$ is scheduled with $u_j$, they are together
with $\theta_{j,i,-,C_1}$, which means the negative literal $\neg
z_i$ appears in $c_j$, and thus $c_j$ is satisfied.

Consider each $(z_i\vee \neg z_k)\in C_2$. There is a huge job
$\theta_{i,k,C_2}$. As machine $\theta_{i,k,C_2}$ is satisfied and
truth benevolent, $\eta_{i,1,+}$ and $\eta_{k,1,-}$ on this machine
should be one true and one false according to Lemma \ref{le:truth
benevolent}.

Suppose on machine $\theta_{i,k,C_2}$, $\eta_{i,1,+}$ is false and
$\eta_{k,1,-}$ is true. Notice that there are two jobs,
$\eta_{i,1,+}^T$ and $\eta_{i,1,+}^F$. Since $\eta_{i,1,+}^F$ is on
machine $\theta_{i,k,C_2}$, $\eta_{i,1,+}^T$ should be on machine
$\theta_{i,1,+}$, and thus on this machine the other job is
$\eta_{i,2,+}^F$. This further implies that $\eta_{i,2,+}^T$ and
$\eta_{i,3,+}^F$ are on machine $\theta_{i,2,+}$. Carry on the above
analysis until we reach machine $\theta_{i,1/\delta-2,+}$, and we
know that $\eta_{i,1/\delta-1,+}^F$ is on this machine. Thus on
machine $\theta_{\eta,i,+}$ the two jobs are
$\eta_{i,1/\delta-1,+}^T$ and $v_{i,2}^F$. See Figure
\ref{fig:truthassign} for an illustration.

\begin{figure}[!h]
\begin{center}
\includegraphics[scale=0.4]{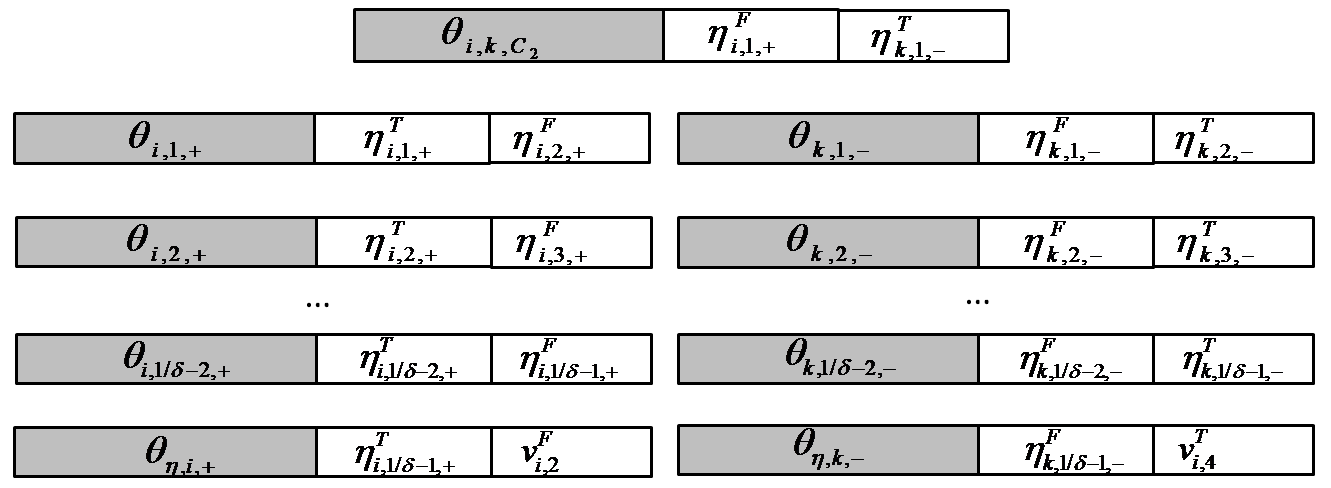}
\caption{truth-assignment}\label{fig:truthassign}
\end{center}
\end{figure}

Similarly, we can show that on machine $\theta_{\eta,k,-}$ the two
jobs are $\eta_{k,1/\delta-1,-}^F$ and $v_{k,4}^T$. Thus, we can
conclude that the variable $z_k$ is false, because otherwise
$v_{k,4}^T$ should be scheduled with $b_k^T$ and $c_k^T$, which is a
contradiction. So the clause $(z_i\vee \neg z_k)$ is satisfied.

Otherwise on machine $\theta_{i,k,C_2}$, the two jobs are
$\eta_{i,1,+}^T$ and $\eta_{k,1,-}^F$. Using the same argument as
before we can show that on machine $\theta_{\eta,i,+}$, the job
$\eta_{i,1/\delta-1,+}$ is false and the job $v_{i,2}$ is true,
while on machine $\theta_{\eta,k,-}$, the job
$\eta_{k,1/\delta-1,-}$ is true and the job $v_{k,4}$ is false.
Thus, the variable $z_i$ is true because otherwise $v_{i,2}^T$
should be scheduled with $b_i^T$ and $d_i^T$, which is a
contradiction. This implies that the clause $(z_i\vee \neg z_k)$ is
satisfied. In both cases, every clause is satisfied, which means
that $I_{sat}$ is satisfiable.
\end{proof}

Recall that given any instance of the 3SAT' problem with $n$
variables, for any $\delta>0$ we construct a scheduling instance
with $O(n/\delta )$ jobs such that it admits a feasible schedule of
makespan $K=O(2^{3/\delta}n^{1+\delta})$ if and only if the given
3SAT' instance is satisfiable. Thus Theorem~\ref{th:sche-makespan}
(and also Theorem~\ref{th:id-sche-linear}) follows directly. We
prove Theorem~\ref{th:id-sche-number}.

\begin{proof}
Suppose the theorem fails, then there exists an exact algorithm for
the restricted scheduling problem that runs in
$2^{O(n^{1-\delta_0})}$ time for some $\delta_0>0$, then we may
simply choose $\delta=\delta_0$ in our reduction. Since $\delta_0$
is some fixed constant, the scheduling problem we construct contains
$O(n)$ jobs with the processing time of each job bounded by
$O(n^{1+\delta_0})$. Then we apply the scheduling algorithm to get
an optimum solution, and it runs in
$2^{O(n^{(1-\delta_0)(1+\delta_0)})}$, i.e.,
$2^{O(n^{1-\delta_0^2})}$ time. Through the makespan of this optimum
solution, we can determine whether the given 3SAT' instance is
satisfiable in $2^{O(n^{1-\delta_0^2})}$ time for some fixed
$\delta_0>0$, resulting a contradiction.
\end{proof}

\section{From $I_{sat}$ to $I_{3dm}$}
\label{appendix:I_3dm}
\subsection{From $I_{sat}$ to $I_{sat}'$}
Suppose we are given an arbitrary $3SAT'$ instance $I_{sat}$ with
$n$ variables. We further apply Tovey's method~\cite{tovey} to
transform $I_{sat}$ into $I_{sat}'$, i.e., we replace each
occurrence of a variable in $I_{sat}'$ with a new variable, and then
add new clauses to enforce that new variables corresponding to the
same original variable are taking the same truth value.

Recall that each variable appears exactly three times in $I_{sat}$,
thus there are in all $3n$ variables in $I_{sat}'$. All the clauses
of $I_{sat}'$ could be divided into two sets, namely $C_1$ and
$C_2$. Every variable appears exactly once in clauses of $C_1$, and
appears twice in clauses of $C_2$. Furthermore, by re-indexing, we
may assume that all the clauses of $C_2$ are $(z_{3k+1}\vee\neg
z_{3k+2})$, $(z_{3k+2}\vee\neg z_{3k+3})$, $(z_{3k+3}\vee\neg
z_{3k+1})$ for $k=0,1,\cdots,n-1$.

We may further assume that $n$ could be divided by $m$ by adding
dummy variables. To see why, suppose $n=qm+r$ with $0<r<m$. Since
$n\ge m$, $q\ge 1$. We could then add additionally $3(m-r)$ dummy
variables, say, $z^d_{3i+1}$, $z^d_{3i+2}$ and $z^d_{3i+3}$ for
$0\le i\le m-r-1$. For these dummy variables, we further introduce
$m-r$ dummy clauses in $C_1$ as $(z^d_{3i+1}\vee z^d_{3i+2}\vee
z^d_{3i+3})$, and $3(m-r)$ clauses in $C_2$ as $(z_{3i+1}^d\vee\neg
z_{3i+1}^d)$, $(z_{3i+2}^d\vee\neg z_{3i+3}^d)$,
$(z_{3i+3}^d\vee\neg z_{3i+1}^d)$ for each $i$.

It is not difficult to verify that $I_{sat}$ is satisfiable if and
only if $I_{sat}'$ is satisfiable.

\subsection{From $I_{sat}'$ to $I_{3dm}$}
We construct an instance of the generalized 3DM problem based on
$I_{sat}'$ (with $3n$ variables). We first construct elements.

We construct two variable elements for each variable $z_i$, i.e., we
construct $w_i$ corresponding to $z_i$ and $\bar{w}_i$ corresponding
to $\neg z_i$. Let $W$ be the set of them. It can be easily seen
that $|W|=6n$. We construct a clause element $s_j\in X$ for each
$c_j\in C_1$.

Recall that all the clauses of $C_2$ could be listed as
$(z_{3i+1}\vee\neg z_{3i+1})$, $(z_{3i+2}\vee\neg z_{3i+3})$,
$(z_{3i+3}\vee\neg z_{3i+1})$ for $i=0,1,\cdots,n-1$. For every $i$,
we construct $a_{3i+1},a_{3i+2},a_{3i+3}\in X$ and
$b_{3i+1},b_{3i+2},b_{3i+3}\in Y$.

This completes the construction of elements and it can be easily
seen that $|X|=3n+m$, and $|Y|=3n$. We construct matchings. For each
variable $z_i$, we construct two matchings of $T_1$, namely $(w_i)$
and $(\bar{w}_i)$.

For each clause $c_j\in C_1$, if the positive literal $z_i\in c_j$,
then we construct $(w_i,s_j)\in T_2$. Else if the negative literal
$\neg z_i\in c_j$, then we construct $(\bar{w}_i,s_j)$. Notice that
$c_j$ might contain two or three literals, thus two or three
matchings of $T_2$ are constructed corresponding to it.

For each $0\le i\le n-1$, 6 matchings of $T_3$ are constructed for
the three clauses $(z_{3i+1}\vee\neg z_{3i+1})$, $(z_{3i+2}\vee\neg
z_{3i+3})$ and $(z_{3i+3}\vee\neg z_{3i+1})$, namely
$(w_{3i+1},a_{3i+1},b_{3i+1})$, $(w_{3i+2},a_{3i+2},b_{3i+2})$,
$(w_{3i+3},a_{3i+3},b_{3i+3})$ and
$(\bar{w}_{3i+1},a_{3i+1},b_{3i+2})$,
$(\bar{w}_{3i+2},a_{3i+2},b_{3i+3})$,
$(\bar{w}_{3i+3},a_{3i+3},b_{3i+1})$.

It can be easily seen that $|T_1|=6n$, $|T_2|=3n$, $|T_3|=6n$. An
exact cover is a subset of matches in which every element appears
once. We prove the following lemma.

\begin{lemma}
\label{le:3sat' to 3dm} $I_{sat}'$ is satisfied if and only if
$I_{3dm}$ admits an exact cover.
\end{lemma}
\begin{proof}
Suppose $I_{sat}'$ is satisfiable, we choose matchings out of $T$ to
form an exact cover.

We know that for each $0\le i\le n-1$, $z_{3i+1}$, $z_{3i+2}$ and
$z_{3i+3}$ are either all true or all false. If they are all true,
then we choose $(\bar{w}_{3i+1},a_{3i+1},b_{3i+2})$,
$(\bar{w}_{3i+2},a_{3i+2},b_{3i+3})$,
$(\bar{w}_{3i+3},a_{3i+3},b_{3i+1})$. Otherwise they are all false,
and $(w_{3i+1},a_{3i+1},b_{3i+1})$, $(w_{3i+2},a_{3i+2},b_{3i+2})$,
$(w_{3i+3},a_{3i+3},b_{3i+3})$ are chosen instead.

Now every element of $Y$ appears exactly once in the matches we
choose currently. Since each clause $c_j\in C_1$ is satisfied, it is
satisfied by at least one variable. We choose the variable that
leads to the satisfaction of $c_j$ (if there are multiple such
variables, we choose arbitrarily one). Suppose this variable is
$z_i$. If $z_i$ is true, then we know the positive literal $z_i\in
c_j$. According to our construction $(w_i,s_j)\in T_2$ and we choose
it. Otherwise $z_i$ is false, and the negative literal $\neg z_i\in
c_j$. Again it follows that $(\bar{w}_i,s_j)\in T_2$ and we choose
it.

Consider the matches we have chosen so far. Every element of $X$ and
$Y$ appears exactly once in these matchings. Moreover, each element
of $W$ appears at most once in these matchings. To see why, notice
that if we choose $(w_i,s_j)\in T_2$, for example, then $z_i$ is
true and we do not choose matchings of $T_3$ that contain $w_i$.
Finally, we choose matchings of $T_1$ to enforce that every element
of $W$ appears exactly once.

On the contrary, suppose there exists an exact cover of $I_{3dm}$,
we prove that $I_{sat}'$ is satisfiable. Consider elements of $X$
and $Y$. For each $0\le i\le n-1$, to ensure that $a_{3i+1}$,
$b_{3i+1}$, $a_{3i+2}$, $b_{3i+2}$ and $a_{3i+3}$, $b_{3i+3}$ appear
once respectively, in the exact cover $T'$ we have to choose either
$(\bar{w}_{3i+1},a_{3i+1},b_{3i+2})$,
$(\bar{w}_{3i+2},a_{3i+2},b_{3i+3})$,
$(\bar{w}_{3i+3},a_{3i+3},b_{3i+1})$, or choose
$(w_{3i+1},a_{3i+1},b_{3i+1})$, $(w_{3i+2},a_{3i+2},b_{3i+2})$,
$(w_{3i+3},a_{3i+3},b_{3i+3})$.

If $(\bar{w}_{3i+1},a_{3i+1},b_{3i+2})$,
$(\bar{w}_{3i+2},a_{3i+2},b_{3i+3})$,
$(\bar{w}_{3i+3},a_{3i+3},b_{3i+1})$ are in $T'$, we set $z_{3i+1}$,
$z_{3i+2}$ and $z_{3i+3}$ to be true. Otherwise we set $z_{3i+1}$,
$z_{3i+2}$ and $z_{3i+3}$ to be false. It can be easily seen that
every clause of $C_2$ is satisfied.

We consider $c_j\in C_1$. Notice that $s_j\in X$ appears once in
$T'$. Suppose the match containing $s_j$ is $(s_j,w_i)$ for some
$i$, then it follows that the positive literal $z_i\in c_j$. The
fact that $w_i$ also appears once implies that $\bar{w}_i$ appears
in $T'\cap T_3$, and thus variable $z_i$ is true and $c_j$ is
satisfied.

Otherwise the matching containing $s_j$ is $(s_j,\bar{w}_i)$ for
some $i$, then similar arguments show that the negative literal
$\neg z_i\in c_j$ and variable $z_i$ is false. Again $c_j$ is
satisfied.
\end{proof}

\section{Dynamic Programming for $Pm||C_{max}$}
\label{appendix:dynamic} We show in this section that the
traditional dynamic programming algorithm for the scheduling problem
runs in $2^{O(\sqrt{ m|I|\log m}+ m\log |I|)}$ time.

Consider the dynamic programming algorithm for the scheduling
problem. Suppose jobs are sorted beforehand as $p_1\le p_2\le
\cdots\le p_n$. We use a vector $(k,t_1,t_2,\cdots,t_{ m})$ to
represent a schedule for the first $k$ jobs where the load of
machine $i$ (i.e., total processing times of jobs on machine $i$) is
$t_i$. Let $ST_k$ be the set of all these vectors that correspond to
some schedules. We could determine $ST_k$ iteratively in the
following way.

Let $ST_0=(0,0,0,\cdots,0)$. For $k\ge 1$, $(k,t_1,t_2,\cdots,t_{
m})\in ST_k$ if there exists some $(k-1,t_1',t_2',\cdots,t_{ m}')\in
ST_{k-1}$ such that for some $1\le i\le  m$, $t_i=t_i'+p_k$, and
$t_j=t_j'$ for $j\neq i$. Since each vector of $ST_{k-1}$ can give
rise to at most $ m$ different vectors of $ST_k$, the computation of
the set $ST_k$ thus takes $O(m|ST_{k-1}|)$ time. Meanwhile, once
$ST_{n}$ is determined, we check each vector of it and select the
one whose makespan is minimized, which also takes $O(m|ST_n|)$ time.
After the desired vector is chosen, we may need to backtrack to
determine how jobs are scheduled on each machine, and this would
take $O(n)$ time.

Thus, the overall running time of the dynamic programming algorithm
mainly depends on the size of the set $|ST_k|$ for $1\le k\le n$. We
have the following lemma.

\begin{lemma}
\label{le:dynamic} $$|ST_k|\le 2^{O(\sqrt{ m|I|\log m}+m\log
|I|)}.$$
\end{lemma}
\begin{proof}
Notice that each vector of $ST_k$ corresponds to some schedule. Let
$J_1,J_2\cdots,J_k$ be the first $k$ jobs with processing times
$1\le p_1\le p_2\le \cdots\le p_k$ and
$\lambda_k=\log_2\prod_{i=1}^k p_i$. Notice that such an indexing of
jobs is only used in the proof, while in the dynamic programming
jobs are in arbitrary order. There are three possibilities.

\vskip 2mm\noindent\textbf{Case 1:} $\log_2 p_1\ge
\sqrt{\lambda_k\log_2m/m}$. Since each vector in $ST_k$ corresponds
to a schedule, we consider all possible assignments of the $k$ jobs.
Each job could be assigned to $ m$ machines, thus there are at most
$m^{k}=2^{k\log_2m}$ different assignments for $k$ different jobs.

Since $1\le p_1\le p_2\le \cdots\le p_k$, we have
$$\lambda_k=\sum_{i=1}^k\log_2 p_i\ge k\sqrt{\lambda_k\log_2m/ m},$$
thus $k\log_2( m+1)\le \sqrt{\lambda_k m\log_2m}$, which implies
that
$$|ST_k|\le 2^{k\log_2m}\le 2^{\sqrt{\lambda_k m\log_2m}}.$$

\vskip 2mm\noindent\textbf{Case 2:} $\log_2 p_k\le
\sqrt{\lambda_k\log_2m/ m}$. Consider any vector of $ST_k$, say,
$(k,t_1,t_2,\cdots,t_{m})$. As $t_i\le kp_k$, there are at most
$(kp_k)^{ m}=2^{ m(\log_2k+\log_2p_k)}$ different vectors. It can be
easily seen that
$$|ST_k|\le 2^{ m(\log_2k+\log_2p_k)}\le 2^{\sqrt{\lambda_k m\log_2m}+ m\log_2k}.$$

\vskip 2mm\noindent\textbf{Case 3:} There exists some $1\le k_0\le
k-1$ such that $\log_2 p_{k_0}\le \sqrt{\lambda_k\log_2m/ m}$ and
$\log_2 p_{k_0+1}\ge \sqrt{\lambda_k\log_2m/ m}$.

Notice that each vector of $ST_k$ corresponds to some schedule.
Given $(k,t_1,t_2,\cdots,t_{m})\in ST_k$, we may let $G_i$ be the
set of jobs on machine $i$. Group $G_i$ can be split into two
subgroups, i.e., jobs belonging to the set
$\{J_1,J_2,\cdots,J_{k_0}\}\cap G_i$ and the set
$\{J_{k_0+1},J_{k_0+2},\cdots,J_{k}\}\cap G_i$. Let $t_i^{(1)}$ be
the total processing time of jobs in the former subgroup and
$t_i^{(2)}$ be the total processing time of jobs in the latter
subgroup. Then the vector $(t_1,t_2,\cdots,t_{m})$ can be expressed
as the sum of two vectors
$$(t_1,t_2,\cdots,t_{m})=(t_1^{(1)},t_2^{(1)},\cdots,t_{m}^{(1)})+(t_1^{(2)},t_2^{(2)},\cdots,t_{m}^{(2)}).$$

Let $ST_k^{(1)}$ and $ST_k^{(2)}$ be the sets of all possible
vectors $(t_1^{(1)},t_2^{(1)},\cdots,t_{m}^{(1)})$ and
$(t_1^{(2)},t_2^{(2)},\cdots,t_{m}^{(2)})$ respectively, then we
know $|ST_k|\le |ST_k^{(1)}|\times|ST_{k}^{(2)}|$. Consider each
vector of $ST_k^{(1)}$, it corresponds to some feasible schedule of
jobs $1$ to $k_0$ over machines. Since $t_i^{(1)}\le k_0p_{k_0}$ and
$\log_2 p_{k_0}\le \sqrt{\lambda_k\log_2m/ m}$, we have
$$|ST_k^{(1)}|\le(k_0p_{k_0})^{m}\le 2^{m\log_2 k_0+\sqrt{\lambda_km\log_2m}}.$$

Consider each vector of $ST_k^{(2)}$, it corresponds to some
feasible schedule of jobs $k_0+1$ to $k$ over machines. To assign
$k-k_0$ different jobs to $m$ machines, there are at most
$m^{k-k_0}=2^{(k-k_0)\log_2m}$ different assignments. Since $\log_2
p_{k_0+1}\ge \sqrt{\lambda_k\log_2m/m}$, we have
$$\lambda_k\ge \sum_{i=k_0+1}^k\log_2 p_i\ge (k-k_0)\sqrt{\lambda_k\log_2m/ m},$$
thus $(k-k_0)\log_2m\le \sqrt{\lambda_k m\log_2m}$, which implies
that
$$|ST_k^{(2)}|\le 2^{(k-k_0)\log_2m}\le 2^{\sqrt{\lambda_km\log_2m}}.$$
Thus, $$|ST_k|\le |ST_k^{(1)}|\times|ST_k^{(2)}|\le 2^{m\log_2
k_0+2\sqrt{\lambda_km\log_2m}}.$$

In any of the above three cases, we always have $$|ST_k|\le
2^{O(\sqrt{m|I|\log m}+ m\log |I|)}.$$
\end{proof}

\end{document}